\documentclass[12pt]{article}
\usepackage[T1]{fontenc}
\usepackage{amsfonts}
\usepackage{amsmath}
\usepackage{amssymb}
\usepackage{amsthm}
\usepackage{bbm}
\usepackage{bm}
\usepackage{mathrsfs}
\usepackage{verbatim}
\usepackage{dsfont}
\usepackage{setspace}
\usepackage{color}
\usepackage{pdfsync}
\usepackage{enumitem}
\usepackage{graphicx}
\usepackage{subfigure}
\usepackage{tikz}
\usetikzlibrary{patterns}
\usepackage{leftindex}
\usepackage[margin=1.0in]{geometry}

\usepackage[pdfborder={0 0 0}]{hyperref}
\hypersetup{
  urlcolor = black,
  pdfauthor = {Steven Campbell, Marcel Nutz},
  pdfkeywords = {Randomization, Transient Price Impact, N-Player Game, Regularization},
  pdftitle = {Randomization in Optimal Execution Games},
  pdfsubject = {Randomization in Optimal Execution Games},
  pdfpagemode = UseNone
}

\usepackage[capitalize]{cleveref} %

\theoremstyle{plain}
\newtheorem{theorem}{Theorem}[section]
\newtheorem{proposition}[theorem]{Proposition}
\newtheorem{lemma}[theorem]{Lemma}

\theoremstyle{definition}
\newtheorem{definition}[theorem]{Definition}
\newtheorem{remark}[theorem]{Remark}

\newtheorem{example}[theorem]{Example}

\newtheorem{assumption}[theorem]{Assumption}
\newtheorem{notation}[theorem]{Notation}
\theoremstyle{remark}

\renewenvironment{thebibliography}[1]{%
\begin{oldthebibliography}{#1}%
\setlength{\baselineskip}{1em}
\linespread{.2}
\small
\setlength{\parskip}{0.2ex}%
\setlength{\itemsep}{.15em}%
}%
{%
\end{oldthebibliography}%
}

\newcommand{\R}{\mathbb{R}}

\newcommand{\cB}{\mathcal{B}}

\newcommand{\noll}{0}

\newcommand{\mykill}[1]{}
\numberwithin{equation}{section}

\begin{document}

\title{\vspace{-1em} Randomization in Optimal Execution Games }
\date{\today}
\author{
  Steven Campbell%
  \thanks{
  Dept.\ of Statistics, Columbia University, sc5314@columbia.edu. Research supported by an NSERC Postdoctoral Fellowship (PDF‑599675-2025).}
  \and
  Marcel Nutz%
  \thanks{
  Depts.\ of Statistics and Mathematics, Columbia University, mnutz@columbia.edu. Research supported by NSF Grants DMS-2407074, DMS-2106056.}
  }
\maketitle \vspace{-1.2em}

\begin{abstract}
We study optimal execution in markets with transient price impact in a competitive setting with $N$ traders. Motivated by prior negative results on the existence of pure Nash equilibria, we consider randomized strategies for the traders and whether allowing such strategies can restore the existence of equilibria. We show that given a randomized strategy, there is a non-randomized strategy with strictly lower expected execution cost, and moreover this de-randomization can be achieved by a simple averaging procedure. As a consequence, Nash equilibria cannot contain randomized strategies, and non-existence of pure equilibria implies non-existence of randomized equilibria. Separately, we also establish uniqueness of equilibria. Both results hold in a general transaction cost model given by a strictly positive definite impact decay kernel and a convex trading cost.
\end{abstract}

\vspace{.3em}

{\small
\noindent \emph{Keywords} Randomization; Optimal Execution; Transient Price Impact; $N$-Player Game

\noindent \emph{AMS 2020 Subject Classification}
91A06; 91A15; 91G10

\noindent \emph{JEL Classification}
G24; 
C62	%
}
\vspace{.6em}

\section{Introduction}

We study optimal execution in a market with transient price impact, meaning that trades dislocate the security's price in the market but this impact diminishes over time according to a decay kernel $G$ (see \cite{CarteaJaimungalPenalva.15,GatheralSchied.13, Webster.23} for background and references). More specifically we are interested in a competitive setting where $N$ traders transact in the same security and are mutually aware of their competitors' trade intention. Early works in this setting include \cite{CarlinLoboViswanathan.07, PedersenBrunnermeier.05, Schoneborn.08, SchonebornSchied.09} whereas recent works with transient price impact include \cite{AbiJaberNeumanVoss.24, CampbellNutz.25a, FuHorstXia.22, NeumanVoss.23, SchiedStrehleZhang.17, SchiedZhang.19,Strehle.17}. To the best of our knowledge, there are no prior works focusing on randomization in optimal execution games.

The starting point of our investigation is a surprising  result of \cite{CampbellNutz.25a, SchiedStrehleZhang.17} which use the  popular Obizhaeva--Wang model~\cite{ObizhaevaWang.13} for price impact; i.e., the decay kernel~$G$ is exponential. In a straightforward generalization of the standard single-player optimal execution task of unwinding a given inventory in a martingale asset over a time interval $[0,T]$, it is shown that the $N$-player game does not admit a Nash equilibrium for any $N\geq2$.\footnote{In \cite{CampbellNutz.25a}, it is discussed how existence is restored when an additional quadratic cost on the trading rate (as in \cite{garleanu.pedersen.16,GraeweHorst.17}) is levied, whereas \cite{SchiedStrehleZhang.17} discretizes time.} In view of the non-existence result, several researchers have asked if allowing for randomized execution strategies would restore existence of a Nash equilibrium. Here randomization means that traders' strategies can depend on privately observed randomization devices. This question is well motivated, as in many other games without a (pure) Nash equilibrium, existence is indeed restored by allowing for randomization. We will show, however, that this is not the case here.

More broadly, we analyze randomized strategies in an $N$-player game with a general impact decay kernel~$G$ and general convex trading costs~$C$ (see \cref{se:setup} for details on the setup). After mathematically formalizing a model with idiosyncratic randomization, we show that randomized strategies are not desirable in our context. Specifically, the main result %
shows that if some trader $i$ considers a randomized strategy $X^{i}$, then while keeping the strategies of the competing traders fixed, replacing $X^{i}$ by a de-randomized version of $X^{i}$ yields a strict improvement in execution cost. In particular, this implies that \emph{a Nash equilibrium cannot contain randomized strategies}. Our analysis further clarifies that this de-randomization can be achieved by the predictable projection onto the market filtration (which does not include the randomization devices). Concretely, this projection simply amounts to averaging the randomized strategy over the scenarios of the randomization device---essentially, taking conditional expectation given the market information.

Mathematically, the result is driven by the strict positive definiteness of the kernel $G$, together with the theory of predictable and dual predictable projections. The proof intuition can be loosely described as a Jensen-type argument exploiting strict convexity: holding the other traders fixed, the transient-impact contribution to a trader's expected execution cost is a strictly convex quadratic functional of her own trading strategy whenever $G$ is strictly positive definite. This suggests that if a trader randomizes idiosyncratically, then averaging her randomized strategy over the randomization scenarios should strictly reduce the impact term, unless the strategy was already non-random. Making this intuition rigorous, however, is not immediate in continuous time because execution costs are built from Stieltjes integrals against trading strategies of finite variation, so one must justify how averaging interacts with the relevant integral terms and with the quadratic form induced by $G$. Our analysis provides this justification and pinpoints the structural features---including linear state dynamics, risk-neutral preferences, and convexity---that explain why, in this class of execution games, idiosyncratic randomization cannot appear in equilibrium.

Separately, we establish the \emph{uniqueness of Nash equilibria}\footnote{Meaning that there is \emph{at most} one equilibrium. As mentioned above, existence can fail depending on the choice of trading costs.} in our general setting. %
This generalizes earlier results (of \cite{SchiedStrehleZhang.17} for exponential decay kernel, $N=2$ traders and quadratic trading costs~$C$, and \cite{CampbellNutz.25a, Strehle.17}  for the same model as~\cite{SchiedStrehleZhang.17} but general~$N$) to a general decay kernel, general convex trading costs, and possibly randomized strategies. Like the result on de-randomization, the uniqueness proof also rests on the strict positive definiteness of the kernel.  

The remainder of this note is organized as follows. \Cref{se:setup} details the setup and notation, focusing on bounded decay kernels. \Cref{se:projections} recalls the predictable and dual predictable projection operators from the theory of stochastic processes.  \Cref{se:no.randomized} contains the main results on de-randomization of randomized strategies and the implications for Nash equilibria, while \cref{se:uniqueness} states the uniqueness of Nash equilibria. \Cref{se:singular} extends all results to the case of unbounded (weakly singular) decay kernels. \Cref{se:conclusion} concludes with remarks about modeling choices. All proofs are reported in \cref{se:proofs}.

\section{Setup}\label{se:setup}

We consider a market where $N\geq1$ agents trade in a single risky asset.\footnote{The case $N=1$ corresponds to the single-player problem where an equilibrium is just an optimal strategy. While some of our results simplify dramatically in that case, all of them remain valid.} We index the traders by $i\in\{1,\dots, N\}$ and denote their inventory processes by $X^{i}=(X^{i}_t)_{t\geq0}$, where $X^i_t$ indicates the number of shares held by trader~$i$ at time~$t$. Each trader $i$ is endowed with initial holdings $X_{0-}^{i}=x^{i}\in\mathbb{R}$ and trades up to the common terminal time $T>0$. Our setup will include problems where full liquidation at time~$T$ is enforced as well as problems where terminal inventory is merely penalized.

\subsection{Randomized Strategies}

Our first goal is to formalize (idiosyncratically) randomized strategies. Fix a filtered probability space $(\Omega^{\noll}, \mathcal{F}^{\noll}, \mathbb{F}^{\noll}=(\mathcal{F}_t^{\noll}), \mathbb{P}^{\noll})$ satisfying the usual conditions. To model randomization, we introduce auxiliary probability spaces 
\((\tilde{\Omega}_{i}, \tilde{\mathcal{F}}_{i}, \tilde{\mathbb{P}}_{i})\), $i=1\dots,N$ and denote by $(\tilde{\Omega}, \tilde{\mathcal{F}}, \tilde{\mathbb{P}})$ the completion of their product $\otimes_{i=1}^{N} (\tilde{\Omega}_{i}, \tilde{\mathcal{F}}_{i}, \tilde{\mathbb{P}}_{i})$. We shall work on the space $(\Omega, \mathcal{F}, \mathbb{P})$ defined as the completion of the product space $(\Omega^{\noll} \times \tilde{\Omega}, \mathcal{F}^{\noll} \otimes \tilde{\mathcal{F}}, \mathbb{P}^{\noll} \otimes \tilde{\mathbb{P}})$.

We define the market subfiltration $\mathbb{F}^M = (\mathcal{F}^M_t)_{t\geq0}$ as the usual augmentation of the filtration $(\mathcal{F}_t^{\noll}\otimes\sigma(\{\emptyset,\tilde{\Omega}\}))_{t\geq0}$.\footnote{The ``usual augmentation'' is the augmentation of the right-continuous version. In this case, $(\mathcal{F}_t^{\noll}\otimes\sigma(\{\emptyset,\tilde{\Omega}\}))_{t\geq0}$ is already right-continuous, so the usual augmentation is just the augmentation by nullsets.} It represents the ``public'' information available to all traders. 
By contrast, agent~$i$'s subfiltration $\mathbb{F}^i=(\mathcal{F}_t^i)_{t \geq 0}$ is defined as the usual augmentation of $(\mathcal{F}_t^{\noll} \otimes \sigma(\tilde{\pi}_{i}))_{t \geq 0}$, where $\tilde{\pi}_{i}:\tilde{\Omega}\to\tilde{\Omega}_i$ is the canonical projection $(\tilde{\omega}_1,\dots,\tilde{\omega}_N)\mapsto \tilde{\omega}_i$.\footnote{$\mathcal{F}_t^{\noll} \otimes \sigma(\tilde{\pi}_{i})$ is a short way of writing $\mathcal{F}_t^{\noll} \otimes \mathcal{E}_{1} \otimes \cdots \otimes\mathcal{E}_{i-1}\otimes\tilde{\mathcal{F}}_{i}\otimes\mathcal{E}_{i+1} \otimes \cdots \otimes\mathcal{E}_{N}$ where $\mathcal{E}_{j}=\sigma(\{\emptyset,\tilde{\Omega}_{j}\})$.} We can think of $\tilde{\pi}_{i}$ as a random number generator whose realization can be an input of agent~$i$'s strategy. 

\begin{notation}\label{notation:marginal.expectations}
  Let $U$ be an integrable random variable on \((\Omega, \mathcal{F},\mathbb{P})\). We denote by $\mathbb{E}\left[U\right]$ the expectation with respect to $\mathbb{P}$. Recall that $\Omega=\Omega^{\noll} \times \tilde{\Omega}$ and  $\mathbb{P}=\mathbb{P}^{\noll} \otimes \tilde{\mathbb{P}}$. We denote by $\tilde{\mathbb{E}}[U]$ the marginal expectation with respect to the second factor,
  \begin{align*}
    \tilde{\mathbb{E}}[U](\omega^{\noll}) &= \int_{\tilde{\Omega}} U(\omega^{\noll},\tilde{\omega}) \tilde{\mathbb{P}}(d\tilde{\omega}).
  \end{align*}
  (We emphasize that $\tilde{\mathbb{E}}[U]$ is a random variable, not a constant.) 
  When $X=(X_t)$ is a process, we abuse notation and write $\tilde{\mathbb{E}}[X]$ for the process $t\mapsto \tilde{\mathbb{E}}[X_t]$.
\end{notation}

 Now that the probability space has been introduced, we can define randomized strategies, the randomization being reflected by the dependence on $\tilde{\omega}_i\in \tilde{\Omega}_i$.

\begin{definition}\label{def:admissible.X}
    We say that a process $X^{i}=(X^{i}_t)_{t\geq0}$ on $\Omega$ is an \emph{admissible}  strategy for trader $i$ if
    \begin{enumerate}
        \item $X^{i}$ is c\`adl\`ag and $\mathbb{F}^i$-predictable,
        \item $X^{i}_{0-}=x^{i}$ and $X^{i}_t$ is constant for $t\geq T$,
        \item the paths $t\mapsto X^{i}_t$ have ($\mathbb{P}$-essentially) bounded total variation.
    \end{enumerate}
    We say that $X$ is a \emph{non-randomized} admissible strategy for trader $i$ if it is also $\mathbb{F}^M$-predictable, whereas it is \emph{strictly randomized} if that is not the case.
\end{definition}

\begin{remark}
    Although we introduce randomization via a private variable sampled at time~$0$, the auxiliary space can be chosen to be arbitrarily rich (for instance, it may encode a sequence of i.i.d.\ uniforms or be function-valued). As a result, our framework still allows traders to randomize their actions in a time-varying manner.    
\end{remark}

\subsection{Price Impact}\label{se:impact}

We assume that the unaffected asset price---the price that would obtain if the $N$ agents did not trade---evolves according to a c\`adl\`ag square-integrable $\mathbb{F}^M$-martingale $P=(P_t)_{t\in[0,T]}$. Note that since the traders' randomization devices are independent of $P$, it remains a martingale in the (larger) private filtrations $\mathbb{F}^i$, $i=1,\dots,N$. To describe the actual, ``affected'' price, we use a general impact decay kernel~$G$; see \cite{GatheralSchied.13} for background and references. For ease of exposition, we first consider bounded kernels---thus avoiding pesky integrability issues---and defer the extension of our results to singular (unbounded) kernels to \cref{se:singular}. Thus, we first impose the following assumption.

\begin{assumption}\label{as:kernel}
  The impact decay kernel $G:\mathbb{R}_+\to \mathbb{R}_+$ is continuous and strictly positive definite in the sense of Bochner; that is, 
  \[
  \int_0^\infty\!\!\int_0^\infty G(|t-s|) dX_s dX_t >0
  \]
  whenever $X:\mathbb{R}_+\to \mathbb{R}$ is a c\`adl\`ag function of finite variation that is not constant.
\end{assumption}

Usually $G$ is chosen to be non-increasing and convex. It is shown in \cite[Proposition~2]{AlfonsiSchiedSlynko.12} that if $G:\mathbb{R}_+\to \mathbb{R}_+$ is non-increasing, convex, and not constant, then the induced kernel is strictly positive definite. Two important examples are the exponential kernel $G(t)=\eta e^{-\lambda t}$ and the truncated power law kernel $G(t)=\eta (1+\lambda t)^{-\gamma}$, where $\eta, \lambda, \gamma>0$ are parameters.

Given a profile $\boldsymbol{X}=(X^{1},\dots,X^{N})$ of admissible strategies for traders $1,\dots,N$, we define the impact process $I=(I_t)_{t\geq0}$ and the affected price $S=(S_t)_{t\geq0}$ by
\begin{equation}\label{eq:def.I.and.S}
    I_t = \int_0^tG(t-s) \sum_{i=1}^NdX_s^{i}, \qquad S_t = P_t + I_t.
\end{equation}
Here and throughout the paper, $\int_{a}^{b}:=\int_{[a,b]}$, meaning that jumps of the integrator at $a$ and $b$ contribute to the integral, whereas 
$\int_{a}^{b-}:=\int_{[a,b)}$ excludes the jump at~$b$. 
Moreover, the ``a.s.'' qualifier is often suppressed.

We define the impact cost associated with $\boldsymbol{X} = (X^{1},\dots,X^{N})$ as follows. For trader $i$, the net proceeds from trading are
\begin{equation}\label{eq:OWcost}\int_0^T S_{t-} dX_t^{i}+\frac{1}{2}\sum_{t\in[0,T]} \Delta S_t  \Delta X_t^{i},
\end{equation}
where $\Delta X_t^{i} :=X^{i}_t - X_{t-}^{i}$. 
Thus continuous trading at time $t$ transacts at the price $S_{t-}$ while a block trade of size $\Delta X_t^{i}$ has an execution price of $S_{t-} +\frac{1}{2}\Delta S_t$. 
This means that trader $i$ obtains the average execution price of all trades happening at~$t$, and is equivalent to randomizing the order in which simultaneous block trades are executed across agents (see, e.g.,~\cite{CampbellNutz.25a} for a more detailed discussion).

Because our problem formulation allows for incomplete liquidation at time~$T$, we must add to the execution costs the change in the marked-to-market value of the holdings over $[0,T]$. This change is
\begin{equation}\label{eq:markToMarket}
    X_{0-}^iP_{0-}-X_T^iP_T
\end{equation}
where we use the unaffected\footnote{This is standard in the literature; see, e.g., \cite{NeumanVoss.22,NeumanVoss.23}. One reason is to ensure that the objective function is convex, which can fail if the reference price is instead taken to be $S_T$.} price $P$ for inventory valuation. If the terminal inventory is constrained to be zero, this accounting amounts to adding a constant to the trader's cost, hence has no effect on their optimal strategy.

\subsection{Additional Trading Costs and Penalizations}

Finally, we introduce additional trading and terminal costs $C(X^i)$. Our exposition allows for general \emph{convex} costs $C$ mapping real-valued c\`adl\`ag functions (realizations of admissible strategies) to $[0,\infty]$. We interpret $C(X)$ pathwise as $C(X)(\omega):=C(t\mapsto X_t(\omega))$ and we assume that $C(X)$ is $\mathcal F$-measurable for all admissible $X$ so that expectations are well defined. The uniqueness result holds in that generality. The results on randomized strategies require an additional Jensen-type inequality (\cref{as:C.no.randomized}) which will be discussed in \cref{se:no.randomized}.

Examples for~$C$ include the quadratic costs often used in the literature. Specifically, quadratic (including zero) trading costs on the trading rate, quadratic or infinite costs on terminal inventory (the latter meaning that the complete liquidation is enforced), and/or quadratic costs on block trades as in \cite{CampbellNutz.25a,SchiedStrehleZhang.17}. 

\begin{example}\label{ex:costsC}
Denoting by $\dot{X}^{i}_t$ the time derivative of $X^{i}_{t}$ (if it exists), let
 \begin{equation} \label{eq:addl.cost.A} 
    C(X^{i}):=
		\frac{\varepsilon}{2}\int_0^T (\dot{X}^{i}_t)^2dt 
		+\frac{1}{2}\sum_{t\in[0,T]} \vartheta_t(\Delta X_t^{i})^2
		+\frac{\varphi}{2} (X_T^{i})^2,
\end{equation}
where $\varepsilon, \vartheta_t \in [0,\infty)$ and $\varphi\in[0,\infty]$. The block cost $\vartheta_t$ is bounded and measurable in $t$. If $\varepsilon>0$, the formula is interpreted as $C(X^{i})=\infty$ if $\dot{X}^{i}$ does not exist, whereas for $\varepsilon=0$ the first term is dropped and nothing is assumed about differentiability. This example includes quadratic (or zero) costs on the trading rate, quadratic (or zero) costs on block trades, quadratic costs on terminal inventory, or a full liquidation constraint at $T$ (for $\varphi=\infty$).
\end{example} 

It is clear that quadratic costs $C$ as in \cref{ex:costsC} are convex. We will verify in \cref{le:C.no.randomized.holds.for.quadratic} that the Jensen-type inequality of \cref{as:C.no.randomized} is also satisfied. A slightly different class, also satisfying our assumptions, are discrete-time models as in \cite{SchiedStrehleZhang.17, Strehle.17}. These can be embedded by a cost~$C$ that charges block costs as in \cref{ex:costsC} on a given grid $\mathbb{T}:=\{0=t_0<t_1<\dots<t_n=T\}$ of trading dates but is infinite for any strategy acting outside the grid (more precisely, the support of its total variation measure is not contained in $\mathbb{T}$).

\subsection{Objective Function and Nash Equilibrium}
 
Combining the additional cost $C$ with the impact cost~\eqref{eq:OWcost} and the value~\eqref{eq:markToMarket} of the terminal inventory, trader $i$ has the following objective function if we fix the actions $\boldsymbol{X}^{-i}=(X^{1},\dots,X^{i-1},X^{i+1},\dots,X^{N})$ of the other players.

\begin{definition}\label{de:objective}
  Given a profile $\boldsymbol{X}=(X^{1},\dots,X^{N})$ of admissible strategies, the objective of trader $i$ is to minimize
\begin{align}\label{eqn:orig.obj} J(X^{i};&\boldsymbol{X}^{-i})
=\mathbb{E}\left[\int_0^T S_{t-} dX_t^{i}+\frac{1}{2}\sum_{t\in[0,T]} \Delta S_t \Delta X_t ^{i} +(X_{0-}^iP_{0-}-X_T^iP_T)+C(X^{i})\right].
\end{align}
\end{definition} 

The next proposition uses the martingale property of the unaffected price and provides a more explicit formula for the objective function.

\begin{proposition}\label{prop:obj.func.rep}
    Given a profile $\boldsymbol{X}=(X^{1},\dots,X^{N})$ of admissible strategies, the objective function $J(X^{i};\boldsymbol{X}^{-i})$ can be expressed as
    \begin{align}\label{eqn:obj.impact.rep}
    J(X^{i};\boldsymbol{X}^{-i})&=\mathbb{E}\left[\int_0^T I_{t-} dX_t^{i}+\frac{1}{2}\sum_{t\in[0,T]}  \Delta I_t \Delta X_t ^{i}+C(X^{i})\right]\\
    &=\mathbb{E}\Bigg[\frac{1}{2}\int_0^T\!\!\int_0^T G(|t-s|)dX_s^{i} dX_t^{i} + \int_0^T\!\!\int_0^{t-} G(t-s) \sum_{j\not=i}dX_s^{j} dX_t^{i} \nonumber \\
    &\quad \quad \quad \quad + \frac{G(0)}{2}\sum_{j\not=i}\sum_{t\in[0,T]} \Delta X_t^{j}\Delta X_t ^{i} +C(X^{i})\Bigg].\label{eqn:obj.schied.rep}
\end{align}
In particular, $J(X^{i};\boldsymbol{X}^{-i})<\infty$ as soon as $\mathbb{E}[C(X^i)]<\infty$.
\end{proposition}

Finally, the definition of Nash equilibrium is the usual one.

\begin{definition} \label{def:nash.equilibrium}
  A profile $\boldsymbol{X}=(X^{1},\dots,X^{N})$ of admissible strategies is a \emph{Nash equilibrium} if $J(X^{i};\boldsymbol{X}^{-i})<\infty$ and
\begin{align*}%
      J(Z^i;\boldsymbol{X}^{-i})\geq J(X^{i};\boldsymbol{X}^{-i})
    \end{align*} 
for all admissible strategies $Z^i$ of trader~$i$, for all $i=1,\dots,N$. %
\end{definition}

\begin{remark}\label{re:info}
We emphasize that, like the previous literature, we start with a full information setup where the unaffected price is known to agents. Under the martingale assumption, the unaffected price ultimately drops out of the objective (see \cref{prop:obj.func.rep}), and hence the strategies do not depend on it.

In our randomized extension, each agent observes the unaffected price together with their own randomization device, but not the competitors' randomization. Consequently, the aggregate order flow, and hence the impact process and affected price, need not be measurable with respect to an individual agent's filtration. For comparison, one may consider an alternative model in which agents' information sets additionally include the affected price (equivalently, the impact process, since the unaffected price is also observed). In that case the objective is expressed in terms of observable quantities (cf.\ \cref{de:objective}), and the game effectively reduces to a common-filtration (full-information) execution game. It is therefore unsurprising that allowing such randomization should not resolve the equilibrium non-existence phenomena established for the corresponding full-information models.

In real markets, the affected price is observed, whereas the unaffected price is not. When the unaffected price is modeled as a martingale but remains latent, observing the affected price may provide partial information about the impact process through filtering, so agents can in principle infer some information about competitors' actions even in the presence of randomization. In the (realistic) regime where the unaffected price is quite noisy and price impact is fairly small relative to that noise, this inference would be quite limited.\footnote{In reality, traders might gain additional information from various sources ranging from observing fills of their own orders in various venues to trade volume data.}  Our present setup, in which the affected price is not measurable in the agents' filtrations, can be interpreted as an idealized high-noise limit in which observing the affected price provides negligible incremental information about impact. We leave a rigorous treatment of partial-information extensions to future work. Results in that direction remain scarce; see \cite{MoallemiParkVanRoy.12} for a model where all impact is permanent and \cite{CasgrainJaimungal.19,CasgrainJaimungal.20} for a mean field model.
\end{remark}

\section{Preliminaries on Projections}\label{se:projections}

The purpose of this section is to recall the properties of the predictable and dual predictable projection operators, and to identify them explicitly in our particular context. The takeaway is that when projecting an admissible strategy of some agent onto the market filtration or the filtration of another agent, the projection is obtained by integrating out agent~$i$'s idiosyncratic randomization. Formally, this corresponds to the rule $\mathbb{E}[f(X,Y)|X]=\mathbb{E}[f(x,Y)]|_{x=X}$ for computing the conditional expectation when $X,Y$ are independent random variables. For background on predictable and dual predictable projections, see for instance 
\cite[Sections~I.2d and I.3b]{JacodShiryaev.03} or \cite[Section VI.2]{DellacherieMeyer.82}. The following two theorems recall the basic properties.

\begin{theorem}[predictable projection]\label{thm:pred.proj}
  Let $Z$ be a bounded $\cB(\R_{+})\otimes\mathcal{F}$-measurable process and let $\mathbb{G}=(\mathcal{G}_t)_{t\geq0}$ be a filtration on $(\Omega,\mathcal{F},\mathbb{P})$ satisfying the usual conditions.
  There exists a unique (up evanescence) $\mathbb{G}$-predictable process ${}^{\rm p^\mathbb{G}} \! Z$, called the predictable projection of $Z$, such that
    \[\mathds{1}_{\{\tau < \infty\}}\,{}^{\rm p^\mathbb{G}} \! Z_\tau=\mathbb{E}[\mathds{1}_{\{\tau < \infty\}}Z_\tau\,\vert\mathcal{G}_{\tau-}]\] \\\ \text{a.s.}
  for every $\mathbb{G}$-predictable stopping time $\tau$.
\end{theorem}

\begin{theorem}[dual predictable projection]\label{thm:dual.pred.proj}
  Let $Z$ be a $\cB(\R_{+})\otimes\mathcal{F}$-measurable process of bounded variation\footnote{Since we are working with a non-trivial initial $\sigma$-field $\mathcal {F}_{0-}$, let us specify that bounded variation processes are understood to have bounded total variation \emph{and} bounded initial value.} and let $\mathbb{G}=(\mathcal{G}_t)_{t\geq0}$ be a filtration on $(\Omega,\mathcal{F},\mathbb{P})$ satisfying the usual conditions. There exists a unique (up evanescence) $\mathbb{G}$-predictable process $Z^{\rm p^\mathbb{G}}$ with bounded variation, called the dual predictable projection of $Z$, such that
    \[\mathbb{E}\left[\int_0^\infty\xi_s dZ^{\rm p^\mathbb{G}}_s\right]=\mathbb{E}\left[\int_0^\infty \xi_s dZ_s\right]\]
    for all bounded $\mathbb{G}$-predictable processes $\xi$. Moreover,
  \begin{equation}\label{eq:dual.predictable.proj.properties}
  \mathbb{E}\left[\int_0^\infty\xi_s dZ^{\rm p^\mathbb{G}}_s\right]=\mathbb{E}\left[\int_0^\infty {}^{\rm p^{\mathbb{G}}}\!\xi_s dZ_s\right]=\mathbb{E}\left[\int_0^\infty {}^{\rm p^{\mathbb{G}}}\!\xi_s dZ_s^{\rm p^\mathbb{G}}\right]
  \end{equation}
for all bounded measurable processes $\xi$.
\end{theorem}

\begin{notation}\label{notation:projections}
  For brevity, the (dual) predictable projection onto the market filtration $\mathbb{F}^M$ will be denoted by ${\rm p}$, whereas the (dual) predictable projection onto trader $i$'s filtration $\mathbb{F}^i$ will be denoted by ${\rm p}^i$,
  \[
    {\rm p}:={\rm p}^{\mathbb{F}^M}, \qquad {\rm p}^i := {\rm p}^{\mathbb{F}^i}.
  \]
\end{notation}

The next two propositions characterize the projections of an admissible strategy of some agent to (a) the market filtration $\mathbb{F}^M$ and (b) the filtration of another agent. In either case, projecting amounts to integrating out the idiosyncratic randomization (see \cref{notation:marginal.expectations}).

\begin{proposition}\label{prop:pred.is.exp}
    Fix $i\in\{1,\dots,N\}$ and let $Z$ be a c\`adl\`ag $\mathbb{F}^i$-predictable process of bounded variation. Then, up to evanescence,
    \[{}^{\rm p}\! Z = Z^{\rm p} = \tilde{\mathbb{E}}[Z].\footnote{In fact, the bounded variation is used only for the dual predictable projection. The identity ${}^{\rm p}\!Z = \tilde{\mathbb{E}}[Z]$ holds for arbitrary bounded $\mathbb{F}^i$-predictable processes~$Z$.}\]
\end{proposition}

\begin{proposition}
    \label{prop:pred.proj.private.info}
    Fix $i,j\in\{1,\dots,N\}$ with $i\not=j$ and let $Z$ be a c\`adl\`ag $\mathbb{F}^j$-predictable process of bounded variation. Then, up to evanescence,
    \[{}^{{\rm p}^i}\! Z = {}^{\rm p}\! Z = \tilde{\mathbb{E}}[Z].\]
\end{proposition}

\section{Non-Existence of Randomized Equilibria}\label{se:no.randomized}

\Cref{prop:pred.is.exp} shows that the projection of a (possibly) randomized admissible strategy~$X^i$ onto the  market filtration is given by~$\tilde{\mathbb{E}}[X^i]$. We record that this de-randomized strategy is again admissible.

\begin{lemma}
    \label{lem:exp.admissible}
     Let $X^i$ be an admissible strategy for player $i$. Then $\tilde{\mathbb{E}}[X^i]$ is a non-randomized admissible strategy for player $i$.
\end{lemma}

Our first proposition shows that an agent is invariant between facing (possibly) randomized strategies of the competitors or their de-randomized projections. More formally, when fixing the strategy profile $\boldsymbol{X}^{-i}$ of the competitors, the expected performance of any strategy $X^i$ against $\boldsymbol{X}^{-i}$ is the same as its performance against the projection $\tilde{\mathbb{E}}[\boldsymbol{X}^{-i}]$ of their strategies onto the market information. %

\begin{proposition}\label{prop:reduced.rep.obj}
    For any admissible strategy profile $\boldsymbol{X}$ and any $i\in\{1,\dots, N\}$, the objective function $J$ satisfies
    \begin{align*}J(X^i;\boldsymbol{X}^{-i}) %
     =J(X^i;\tilde{\mathbb{E}}[\boldsymbol{X}^{-i}]).
    \end{align*}
\end{proposition}

Our main result below uses the following condition on the convex cost functional $C$, which can be viewed as a Jensen-type inequality for marginal expectations. Such inequalities hold in standard Banach-space settings under additional assumptions, for instance for Bochner expectations and continuous convex functionals; cf.~\cite{to1975generalized}. In our setup, however, we do not impose a specific topological vector space structure on the class of c\`adl\`ag finite-variation strategies, and we allow general extended-valued path functionals $C$ (e.g., to encode hard constraints such as discrete-time trading). In this generality, convexity of $C$ alone does not imply Assumption~\ref{as:C.no.randomized} (see \cref{app:counterexample} for a counterexample).

\begin{assumption}\label{as:C.no.randomized}
  If $X$ is an admissible strategy for some trader $i$, then
  \[\mathbb{E}\left[C(X)\right]\geq \mathbb{E}\left[C\left(\tilde{\mathbb{E}}[X]\right)\right].\]
\end{assumption}

Next, we verify that \cref{as:C.no.randomized} is satisfied for the quadratic costs that are predominantly used in the literature.

\begin{lemma}\label{le:C.no.randomized.holds.for.quadratic}
    \Cref{as:C.no.randomized} holds for the quadratic costs $C$ specified in \cref{ex:costsC}.
\end{lemma}

We can now state our first main result. It shows that de-randomizing a randomized strategy by projecting leads to an improved performance while holding competitors' strategies fixed. In particular, this implies the non-existence of (strictly) randomized equilibria.

\begin{theorem}\label{thm:non.existence.random.eq}
    Let $C$ satisfy \cref{as:C.no.randomized} and let $\boldsymbol{X}=(X^{1},\dots,X^{N})$ be an admissible strategy profile where $X^i$ is strictly randomized. Then $Z^i:=\tilde{\mathbb{E}}[X^i]$ is a non-randomized admissible strategy for trader $i$ and 
    \[J(X^i; \boldsymbol{X}^{-i}) > J(Z^i; \boldsymbol{X}^{-i}).\]
    That is, a strictly randomized strategy $X^i$ can be strictly improved by unilaterally deviating to the non-randomized strategy $Z^i = \tilde{\mathbb{E}}[X^i]$. In particular, a Nash equilibrium profile cannot contain strictly randomized strategies. 
\end{theorem}

The proof shows that this result is driven by the strict positive definiteness of the impact decay kernel. The assumption on $C$, which is a non-strict inequality, merely ensures that $C$ does not counteract that property of the kernel.

\section{Uniqueness of Nash Equilibria}\label{se:uniqueness}

We show the uniqueness of Nash equilibria in our general framework.\footnote{To be clear, \cref{as:C.no.randomized} is not in force in this section.} The first result in this direction is due to \cite{SchiedStrehleZhang.17}, where uniqueness is shown for the exponential decay kernel (i.e., the Obizhaeva--Wang model), $N=2$ traders and a quadratic cost $C$ that is a special case of \cref{ex:costsC}. This result was extended to similar games with an arbitrary finite number~$N$ of traders in \cite{CampbellNutz.25a,Strehle.17}. A different uniqueness argument, based on Fredholm equations of the second kind, was given in \cite{AbiJaberNeumanVoss.24} for linear-quadratic models that are regularized by a quadratic cost on the trading rate. Here, we establish a general version of the uniqueness property including general convex costs~$C$, a general decay kernel, and (possibly) randomized strategies. We emphasize that the theorem merely asserts uniqueness; existence need not hold under the present assumptions (see, e.g., \cite{CampbellNutz.25a}).

\begin{theorem}\label{thm:nash.eq.unique}
    There is at most one Nash equilibrium. %
\end{theorem}

In particular, this implies that there exists at most one pure Nash equilibrium. Following  the initial argument of \cite{SchiedStrehleZhang.17}, the proof in~\cref{se:proofs.uniqueness} rests on the strict positive definiteness of the decay kernel. The cost $C$ is assumed to be convex but not necessarily strictly convex ($C\equiv0$ is allowed). Again, the convexity assumption makes sure that the cost does not counteract the properties of the kernel, but on its own clearly does not yield uniqueness.

\section{Extension to Singular Kernels}\label{se:singular}

In this section we extend our results to impact decay kernels with $G(0+)=\infty$. See \cite{GatheralSchied.13} for the significance of such kernels. The most important example is the power-law kernel $G(t)=t^{\gamma}$ where $\gamma\in(0,1)$; cf.\ \cite{Gatheral.10}. We put ourselves in the framework of~\cite{GatheralSchiedSlynko.12} for (weakly) singular kernels and impose the following throughout this section.

\begin{assumption}[singular kernel]\label{as:singular.kernel}
  The function $G:\mathbb{R}_+\to [0,\infty]$ is convex and non-increasing, continuous and finite on $(0,\infty)$ with $\int_{0}^{1} G(t) dt<\infty$, and $G(0)=\lim_{t\downarrow0} G(t)=\infty$.
\end{assumption}

Note that, for convenience of exposition, we have made $G(0)=\infty$ part of the assumption: we focus on the singular case, rather than including the bounded case of \cref{as:singular.kernel} whose results were already stated. The definition of admissibility needs to be amended as follows.

\begin{definition}\label{def:admissible.X.singular}
    A process $X^{i}=(X^{i}_t)_{t\geq0}$ is an \emph{admissible} strategy for trader $i$ if it satisfies \cref{def:admissible.X} and in addition
				\begin{equation}\label{eq:admissible.integrability}
                \mathbb{E}\left[\int_0^\infty\!\!\int_0^\infty G(|t-s|) d|X^{i}|_sd|X^{i}|_t\right]<\infty.
				\end{equation}
\end{definition}

While~\eqref{eq:admissible.integrability} was automatic for bounded~$G$ (as $X^{i}$ is of bounded variation), this additional condition is needed to ensure that the quantities of interest are well-defined when~$G$ is unbounded. The next remark emphasizes that in the present setting with $G(0)=\infty$, \emph{admissible strategies cannot have jumps}; that is, there are no block trades.

\begin{remark}\label{rk:no.jumps}
  The integral $\int_0^T \int_0^T G(|t-s|) d|X|_s d|X|_t$ is necessarily infinite if~$X$ has a jump, as $\int_0^T \int_0^T G(|t-s|) d|X|_s d|X|_t\geq G(0)|\Delta X_{t}|^{2}$ for any~$t$ and $G(0)=\infty$. Thus, in view of~\eqref{eq:admissible.integrability}, admissible strategies cannot have jumps in the present setting. In financial terms, block trades have infinite price impact and hence infinite execution cost.
\end{remark}

As in the previous sections, our main results will be driven by the strict positive definiteness of the kernel. Let us record that strict positive definiteness indeed holds. As detailed in the proof, this property can be seen from a Fourier-type representation shown in~\cite{GatheralSchiedSlynko.12}.

\begin{lemma}\label{le:singular.posdef}
  Let $G$ satisfy \cref{as:singular.kernel}. Then $G$ is strictly positive definite; i.e.,
  \[
  \int_0^T\!\!\int_0^T G(|t-s|) dX_s dX_t >0
  \]
  whenever $X:[0,T]\to \mathbb{R}$ is a c\`adl\`ag function of finite variation that is not constant and satisfies $\int_0^T\!\!\int_0^T G(|t-s|) d|X|_s d|X|_t<\infty$.
\end{lemma}

The following approximation will enable us to exploit our results for bounded kernels.

\begin{remark}[monotone approximation]\label{rk:kernel.approx}
  Let $G$ satisfy \cref{as:singular.kernel}. Then given any $n>0$ such that $G$ is not constant on $[1/n,\infty)$, the function $G^n(t):=G(t+1/n)$, $t\in\R_{+}$ is a bounded kernel satisfying \cref{as:kernel} (by the sufficient condition mentioned below \cref{as:kernel}). %
  Moreover, we have the monotone limit $G^n(t)\uparrow G(t)$ as $n\uparrow\infty$. %
  This will be useful for monotone convergence arguments.
\end{remark} 

The objective function~$J$ is the same as in \cref{se:setup}, with the following caveats. First, the expressions involving jump terms can be dropped by \cref{rk:no.jumps}. %
Second, in contrast to the case of bounded kernels, it may not be obvious that various integrals are well defined, and some steps in the proof of \cref{prop:obj.func.rep} need additional explanations. We provide those in \cref{se:proofs.singular}, and we also verify by elementary arguments that~\eqref{eq:admissible.integrability} and positive definiteness imply that all expressions are well defined. We thus have the following simplified expression for \cref{prop:obj.func.rep}.

\begin{proposition}\label{prop:singular.obj.func.rep}
    Given a profile $\boldsymbol{X}=(X^{1},\dots,X^{N})$ of admissible strategies, the objective function $J(X^{i};\boldsymbol{X}^{-i})$ can be expressed as
    \begin{align}%
    J(X^{i};\boldsymbol{X}^{-i})&=\mathbb{E}\left[\int_0^T I_{t-} dX_t^{i}+C(X^{i})\right]\nonumber \\
    &=\mathbb{E}\Bigg[\frac{1}{2}\int_0^T\!\!\int_0^T G(|t-s|)dX_s^{i} dX_t^{i} + \int_0^T\!\!\int_0^{t-} G(t-s) \sum_{j\not=i}dX_s^{j} dX_t^{i} +C(X^{i})\Bigg].\label{eqn:singular.obj.schied.rep}
\end{align}
Moreover, $|J(X^{i};\boldsymbol{X}^{-i})|<\infty$ as soon as $\mathbb{E}[C(X^i)]<\infty$.
\end{proposition}

A technical difficulty in the present setup is that the impact process $I_t$ may be infinite for some exceptional set of~$t$, even when strategies are admissible. This does not contradict the preceding statements because the objective function only includes certain integrals of the impact process, not the impact process at a fixed time. Thus, strictly speaking, there is no need to address the difficulty. Let us mention, however, the following result of \cite[Proposition~2.27]{GatheralSchiedSlynko.12}.

\begin{remark}\label{rk:impact.finite.quasi}
   Let $X:[0,T]\to\mathbb{R}$ be c\`adl\`ag function of finite variation with $\int_0^T\!\!\int_0^T G(|t-s|) d|X|_s d|X|_t<\infty$. Then $\int_0^t G(t-s) dX_s$ is well defined and finite for \emph{quasi} every $t\in[0,T]$, meaning that the set of exceptional~$t$ is $\nu$-null for every finite Borel measure $\nu$ satisfying $\iint G(|t-s|)\nu(ds)\nu(dt)<\infty$. 
\end{remark} 

We may apply this with $X$ being a path of an admissible strategy and $\mu$ being the total variation measure of the path of another admissible strategy. But, as mentioned, the elementary results in \cref{se:proofs.singular} suffice for our purpose. 

We can now formally state that all our results extend to the singular setting.

\begin{theorem}\label{thm:extensions.singular}
  Let $G$ satisfy \cref{as:singular.kernel} and define admissibility as in \cref{def:admissible.X.singular}. Then \cref{lem:exp.admissible}, \cref{prop:reduced.rep.obj},
  \cref{thm:non.existence.random.eq} and \cref{thm:nash.eq.unique} hold as stated.
\end{theorem}

We emphasize that this is purely a de-randomization and uniqueness result; it makes no assertions about (non)existence of a deterministic equilibrium. Existence results are available in \emph{regularized} execution games, typically using additional quadratic costs (e.g., on the trading rate), and for broad classes of decay kernels; see, for instance, \cite{AbiJaberNeumanVoss.24} and the references therein. Absent regularization, the conclusions may be different for \emph{singular} kernels with $G(0)=\infty$. In this regime block trades are inadmissible, which removes one of the mechanisms behind the equilibrium nonexistence in classical Obizhaeva--Wang type models. We are unaware of work in this direction. %

\section{Conclusion}\label{se:conclusion}

In our model of an optimal execution game with transient price impact, we have shown that Nash equilibria are unique and do not contain randomized strategies. Our starting point was the non-existence of pure Nash equilibria in certain optimal execution games with full information, and the question whether allowing for randomization restores existence. We have chosen a direct extension of the full information setup predominant in the literature. While in some other games, such an extension enables Nash equilibria, that is not the case here. Our setting is the strongest form of randomization because traders cannot see the impact of each other's actions, and thus gives a conservative answer to the initial question. A natural direction for future work is a systematic treatment of partial-information variants in which the affected price is observed while the unaffected price is latent, so that agents can potentially infer impact and competitors' activity through filtering.

\section{Proofs}\label{se:proofs}

\subsection{Notation}

Recall that we write $\int_a^b H_{t}dX_{t}$ to denote the integral over $[a,b]\subset \mathbb{R}$ for suitable integrands~$H$ and integrators~$X$. In particular, $\int_a^b H_{t}dX_{t}$ includes the contribution $H_{a}\Delta X_{a}$ of a jump in $X$ at $t=a$ and similarly at $t=b$. Whereas, $\int_a^{b-} H_{t}dX_{t}$ is the integral over $[a,b)$, excluding the jump at $b$. 

The subsequent subsections, except for the last one, all consider a (bounded) kernel $G$ satisfying \cref{as:kernel}. The last subsection deals with the extension to singular kernels, meaning that $G$ satisfies \cref{as:singular.kernel}.

\subsection{Proofs for \cref{se:setup} (Objective Function)}\label{se:proofs.setup}

\begin{proof}[Proof of \cref{prop:obj.func.rep}]
Expanding \eqref{eqn:orig.obj} gives
\begin{align*}J(X^{i};\boldsymbol{X}^{-i})&=\mathbb{E}\bigg[\int_0^T P_{t-} dX_t^{i}+\frac{1}{2}\sum_{t\in[0,T]}  \Delta P_t \Delta X_t^{i} +\int_0^T I_{t-} dX_t^{i} + \frac{1}{2}\sum_{t\in[0,T]} \Delta I_t \Delta X_t ^{i}\\
&\quad \quad \quad \quad +(X_{0-}^iP_{0-}-X^{i}_TP_T)+C(X^{i})\bigg].
\end{align*}
Using integration by parts,
\begin{align*}
    \int_0^T P_{t-} dX_t^{i}  &= P_TX_T^i-P_{0-}X_{0-}^i-\int_0^T X_{t-}^i dP_t -[X^i,P]_T
\end{align*}
and thus
\begin{align*}J(X^{i};\boldsymbol{X}^{-i})&=\mathbb{E}\bigg[-\int_0^T X^{i}_{t-} dP_t - [X^{i},P]_T+\frac{1}{2}\sum_{t\in[0,T]}  \Delta P_t \Delta X_t^{i} \\
&\quad \quad \quad \quad +\int_0^T I_{t-} dX_t^{i} + \frac{1}{2}\sum_{t\in[0,T]} \Delta I_t \Delta X_t^{i}  +C(X^{i})\bigg].
\end{align*}
As $P$ is a square-integrable $(\mathbb{F}^i,\mathbb{P})$-martingale and $X^i$ is a bounded predictable process, both $\int_0^t X_{s-}^i dP_s$ and $[X^i,P]_t=\int_0^t\Delta X^i_sdP_s$ are true $(\mathbb{F}^i,\mathbb{P}) $-martingales, hence have vanishing expectation. Moreover,
\[\frac{1}{2}\sum_{t\in[0,T]} \Delta X_t^{i} \Delta P_t=\frac{1}{2}[X^{i},P]_T,\]
so its expectation also vanishes. This proves the claimed representation~\eqref{eqn:obj.impact.rep}.

For the remaining representation~\eqref{eqn:obj.schied.rep}, we insert the definition~\eqref{eq:def.I.and.S} of~$I$ into \eqref{eqn:obj.impact.rep} to get
\begin{align}\label{eqn:obj.intermediate.rep}J(X^{i};&\boldsymbol{X}^{-i})
= \mathbb{E}\bigg[C(X^{i})+ \int_0^T\!\!\int_0^{t-} G(t-s) \sum_{j=1}^NdX_s^{j} dX_t^{i} + \frac{G(0)}{2}\sum_{j=1}^N\sum_{t\in[0,T]}  \Delta X_t^{j} \Delta X_t ^{i}\bigg].%
\end{align}
Splitting the integral and applying Fubini's theorem to interchange the order of integration,
\begin{align*}\int_0^T\!\!\int_0^{t-} &G(t-s) dX_s^{i} dX_t^{i}\\
&=\frac{1}{2}\int_0^T\!\!\int_0^{t-} G(t-s) dX_s^{i} dX_t^{i}+\frac{1}{2}\int_0^T\!\!\int_0^{t-} G(t-s) dX_s^{i} dX_t^{i}\\
&=\frac{1}{2}\int_0^T\!\!\int_0^{t-} G(t-s) dX_s^{i} dX_t^{i}+\frac{1}{2}\int_0^T\!\!\int_{s}^T G(t-s) dX_t^{i} dX_s^{i}-\frac{G(0)}{2}\sum_{s\in[0,T]}(\Delta X^i_s)^2\\
&=\frac{1}{2}\int_0^T\!\!\int_0^T G(|t-s|) dX_s^{i} dX_t^{i}-\frac{G(0)}{2}\sum_{t\in[0,T]}(\Delta X^i_t)^2.
\end{align*}
Substituting this expression into \eqref{eqn:obj.intermediate.rep} completes the proof of~\eqref{eqn:obj.impact.rep}. Finiteness of the objective when $\mathbb{E}[C(X^{i})]<\infty$ is clear from~\eqref{eqn:obj.impact.rep}. 
\end{proof}

\subsection{Proofs for \cref{se:projections} (Projections)}\label{se:proofs.projections}

\begin{lemma}\label{lem:filtrationRC}
  The filtration $(\mathcal{F}_t^{\noll}\otimes \sigma(\{\emptyset,\tilde{\Omega}\}))_{t\geq0}$ is right continuous and hence its augmentation is the market filtration $\mathbb{F}^{M}$.
\end{lemma}

\begin{proof}
  It is elementary to verify the right-continuity. The second claim follows as $\mathbb{F}^{M}$ was defined as the usual augmentation of $(\mathcal{F}_t^{\noll}\otimes \sigma(\{\emptyset,\tilde{\Omega}\}))_{t\geq0}$.
\end{proof} 

We record the following  consequence of \cref{lem:filtrationRC} for ease of reference.

\begin{lemma}\label{lem:mkt.cond.exp.simplification}
  Let $t\geq0$ and let $U$ be an integrable $\mathcal{F}$-measurable random variable. Then $\mathbb{E}\left[U\big|\mathcal{F}_t^{\noll}\otimes \sigma(\{\emptyset,\tilde{\Omega}\})\right]=\mathbb{E}\left[U\vert \mathcal{F}^M_t\right]$ $\mathbb{P}$-a.s.
\end{lemma}

The next lemma formalizes that conditioning a random variable from agent $i$'s filtration to the market filtration is equivalent to integrating out the randomization.

\begin{lemma}\label{lem:cond.exp.mkt}
    Let $t\geq0$ and let $U$ be an integrable $\mathcal{F}^i_t$-measurable random variable for some $i\in\{1,\dots, N\}$. Then 
    $\tilde{\mathbb{E}}[U]=\mathbb{E}[U|\mathcal{F}_{t}^M]$ $\mathbb{P}$-a.s.
\end{lemma}

\begin{proof}
    By \cref{lem:mkt.cond.exp.simplification} it suffices to show that $\tilde{\mathbb{E}}[U]$ is a version of $\mathbb{E}[U|\mathcal{F}_{t}^{\noll}\otimes\sigma(\{\emptyset,\tilde{\Omega}\})]$.  The measurability of marginal expectations on completed product spaces entails that $\omega^{\noll}\mapsto\tilde{\mathbb{E}}\left[U\right](\omega^{\noll})$ is $\mathcal{F}_t$-measurable, and hence that $(\omega^{\noll},\tilde{\omega}) \mapsto\tilde{\mathbb{E}}\left[U\right](\omega^{\noll},\tilde{\omega})$ is $\mathcal{F}_{t}^{\noll}\otimes\sigma(\{\emptyset,\tilde{\Omega}\})$-measurable. Let $A\in\mathcal{F}_{t}^{\noll}\otimes\sigma(\{\emptyset,\tilde{\Omega}\})$, so that $A=A' \times \tilde{\Omega}$ for some $A'\in\mathcal{F}_{t}^{\noll}$. By the definition of the product measure $\mathbb{P}$, we have
$\int_{A} U d\mathbb{P} =\int_{A'} \int_{\tilde{\Omega}}U d\tilde{\mathbb{P}}d\mathbb{P}^{\noll}=\int_{A'} \tilde{\mathbb{E}}\left[U\right]d\mathbb{P}^{\noll}=\int_{A}\tilde{\mathbb{E}}\left[U\right]d\mathbb{P}. 
$
\end{proof}

We can now detail the proof of \cref{prop:pred.is.exp}, which is a version of \cref {lem:cond.exp.mkt} for processes.

\begin{proof}[Proof of \cref{prop:pred.is.exp}]
    If $Z$ is an elementary $\mathbb{F}^i$-predictable process then it takes the form
    \[Z = Z_{0-}\mathds{1}_{0-}(\cdot)+ Z_{0}\mathds{1}_{0}(\cdot)+\sum_{k=1}^d Z_{s_k}\mathds{1}_{(s_k,t_k]}(\cdot)\]
    for some $s_k<t_k$, $\mathcal{F}_{s_k}^i$-measurable $Z_{s_k}$, and $\mathcal{F}_{0-}^i$-measurable $Z_{0-},Z_0$. Taking expectations with respect to $\tilde{\mathbb{P}}$ and invoking \cref{lem:cond.exp.mkt}, it is clear that $\tilde{\mathbb{E}}[Z]$ is an elementary  $\mathbb{F}^M$-predictable process. By a monotone class argument it follows that $\tilde{\mathbb{E}}[Z]$ is $\mathbb{F}^M$-predictable also for a general bounded predictable processes~$Z$. %

    Let $Z$ be bounded and $\mathbb{F}^i$-predictable; we show that ${}^{\rm p} Z = \tilde{\mathbb{E}}[Z]$. Since $\tilde{\mathbb{E}}[Z]$ is predictable, $\tilde{\mathbb{E}}[Z_{\tau}]$ is $\mathcal{F}^M_{\tau-}$-measurable for all $\mathbb{F}^M$-predictable times $\tau$.\footnote{Our notation is ambiguous, but note that the integral $\tilde{\mathbb{E}}[Z_{\tau}]$ indeed coincides the the evaluation of the process $t\mapsto \tilde{\mathbb{E}}[Z_{t}]$ at $t=\tau$ when $\tau$ is an $\mathbb{F}^M$-stopping time.}
    Thus, we may check the definition of the conditional expectation. Recall that by definition,
    \begin{align*}\mathcal{F}_{\tau-}^M &=\sigma\left(\left\{B \cap \{t<\tau\}: t\geq0, B\in \mathcal{F}_t^M\right\} \cup\mathcal{F}_{0-}^M\right).
    \end{align*}
    All of these generating sets are necessarily of the form $(A'\times \tilde{\Omega}) \cup N$ for a $\mathbb{P}$-null set $N$, and this extends to the $\sigma$-field. For any $A=(A'\times \tilde{\Omega}) \cup N\in\mathcal{F}_{\tau-}^M$, 
    we argue as in the proof of \cref{lem:cond.exp.mkt} that $\int_{A} Z_\tau d\mathbb{P} = \int_{A}\tilde{\mathbb{E}}\left[Z_\tau\right]d\mathbb{P}$.
    We conclude that 
    $\mathbb{E}[Z_\tau|\mathcal{F}_{\tau-}^M] =\tilde{\mathbb{E}}[Z_\tau]$ and so,
    \[\mathbb{E}[1_{\{\tau < \infty\}} Z_\tau\vert\mathcal{F}_{\tau-}^M]=1_{\{\tau < \infty\}} \tilde{\mathbb{E}}[Z_\tau].\]
    By the uniqueness up to evanescence in \cref{thm:pred.proj}, ${}^{\rm p} Z = \tilde{\mathbb{E}}[Z]$. 

We turn to the claim that $Z^{\rm p}=\tilde{\mathbb{E}}[Z]$ when $Z$ is c\`adl\`ag and of bounded variation. By the above, $\tilde{\mathbb{E}}[Z]$ is a $\mathbb{F}^M$-predictable process of bounded variation. In particular, $\tilde{\mathbb{E}}[Z]=\tilde{\mathbb{E}}[Z]^{\rm p}$. To show that $Z^{\rm p}=\tilde{\mathbb{E}}[Z]$, it then suffices (by \cite[Theorem VI.75(b)]{DellacherieMeyer.82}) to check  that
$$
  \mathbb{E}\left[Z_{0}| \mathcal{F}^M_{0-}\right]=\mathbb{E}\left[ \tilde{\mathbb{E}}[Z_{0}] \big| \mathcal{F}^M_{0-}\right], \qquad
\mathbb{E}\left[Z_{\infty}-Z_{t}\big| \mathcal{F}^M_{t}\right]=\mathbb{E}\left[ \tilde{\mathbb{E}}[Z_{\infty}] - \tilde{\mathbb{E}}[Z_{t}] \big| \mathcal{F}^M_{t}\right], \quad t\geq0.
$$
We prove the second part; the first is similar. Fix $t\geq0$. To simplify notation, let $\zeta=Z_{\infty}-Z_{t}$. %
Moreover, note that it suffices to show
\begin{equation}\label{eq:F.proof.cond}
  \mathbb{E}\left[\zeta| \mathcal{F}^M_{\infty}\right]=\tilde{\mathbb{E}}[\zeta],
\end{equation}
as an application of the tower property then yields the claim. To show~\eqref{eq:F.proof.cond}, note that by
construction of $\mathbb{F}^i$ we almost surely have $\zeta(\omega^{\noll},\tilde{\omega})=\bar{\zeta}(\omega^{\noll}, \tilde{\omega}_{i})$
for some measurable function~$\bar{\zeta}$. 
Using the product form of $\mathbb{P}$,
\begin{align*}
  \mathbb{E}[ \zeta \vert\mathcal{F}^M_{\infty}] (\omega^{\noll},\tilde{\omega})
  = \int \bar{\zeta}(\omega^{\noll},\tilde{\omega}_{i})
  \tilde{\mathbb{P}}_{i}(d\tilde{\omega}_{i}) 
  =  \tilde{\mathbb{E}}[\zeta] (\omega^{\noll}),
\end{align*} 
completing the proof.
\end{proof}

Finally, we prove the analogue for projecting onto the filtration of another agent.

\begin{proof}[Proof of \cref{prop:pred.proj.private.info}]
By \cref{prop:pred.is.exp}, %
$\tilde{\mathbb{E}}[Z]$ is an $\mathbb{F}^M$-predictable process. In particular, $\tilde{\mathbb{E}}[Z]$ is predictable in the filtration~$\mathbb{F}^i\supset \mathbb{F}^M$ and it suffices to check that
    \[\mathds{1}_{\{\tau < \infty\}}\tilde{\mathbb{E}}[Z_t]|_{t=\tau}=\mathbb{E}[\mathds{1}_{\{\tau < \infty\}}Z_\tau\,\vert\mathcal{F}^i_{\tau-}]\]
    for all $\mathbb{F}^i$-predictable times $\tau$. 
    By construction of $\mathbb{F}^i$ and $\mathbb{F}^j$ we almost surely have 
    \[
    \tau(\omega^{\noll},\tilde{\omega})=\bar{\tau}(\omega^{\noll}, \tilde{\omega}_{i}), \qquad  \mathds{1}_{\{t < \infty\}}Z_{t}(\omega^{\noll},\tilde{\omega})=\bar{Z}(t,\omega^{\noll},\tilde{\omega}_{j})
    \]
    for some measurable functions $\bar{\tau},\bar{Z}$. Hence $\mathds{1}_{\{\tau < \infty\}}Z_\tau$ is the composition
    \[
      (\mathds{1}_{\{\tau < \infty\}}Z_\tau)(\omega^{\noll},\tilde{\omega})= \bar{Z}(\bar{\tau}(\omega^{\noll},\tilde{\omega}_{i}),\omega^{\noll}, \tilde{\omega}_{j})
    \]
    and by the product form of $\mathbb{P}$,
		\begin{align*}
      \mathbb{E}[\mathds{1}_{\{\tau < \infty\}}Z_\tau\,\vert\mathcal{F}^i_{\infty}] (\omega^{\noll},\tilde{\omega})
      &= \int \bar{Z}(\bar{\tau}(\omega^{\noll},\tilde{\omega}_{i}),\omega^{\noll},\tilde{\omega}_{j})
      \tilde{\mathbb{P}}_{j}(d\tilde{\omega}_{j}) \\
      &= \big(\mathds{1}_{\{\tau < \infty\}} \tilde{\mathbb{E}}[Z_t]|_{t=\tau}\big) (\omega^{\noll},\tilde{\omega}).
		\end{align*} 
		As we already know that the right hand side is measurable with respect to $\mathcal{F}^i_{\tau-}\subset \mathcal{F}^i_{\infty}$, the claim follows by the tower property of conditional expectation. 
\end{proof}

\subsection{Proofs for \cref{se:no.randomized} (De-Randomization)}\label{se:proofs.no.randomized}

\subsubsection{Admissibility of $\tilde{\mathbb{E}}[X]$ (\Cref{lem:exp.admissible})}

In what follows we denote by $TV(X;[a,b])$ the total variation of the path $t\mapsto X_t$ on $[a,b]$. Our first result is straightforward.

\begin{lemma}\label{lem:TV.ineq} 
$TV(\tilde{\mathbb{E}}[X];[a,b])\leq \tilde{\mathbb{E}}[TV(X;[a,b])]$ for any $0\leq a\leq b\leq \infty$
\end{lemma}

\begin{proof}
    For any partition $\Pi=\{t_0,\dots,t_n\}$ of $[a,b]$,
\begin{align*}
    \sum_{i=0}^{n-1}|\tilde{\mathbb{E}}[X_{t_{i+1}}]-\tilde{\mathbb{E}}[X_{t_i}]|&\leq \tilde{\mathbb{E}}\left[\sum_{i=0}^{n-1}|X_{t_{i+1}}-X_{t_i}|\right]\leq \tilde{\mathbb{E}}[TV(X;[a,b])].
\end{align*}
Taking the supremum over partitions $\Pi$ yields the claim.
\end{proof}
 
\begin{proof}[Proof of \Cref{lem:exp.admissible}]
As $\tilde{\mathbb{E}}[X^i]={}^{\rm p}X^i$ (cf.\ \cref{prop:pred.is.exp}), it is non-randomized and $\mathbb{F}^M$-predictable, hence $\mathbb{F}^i$-predictable. Using \cref{lem:TV.ineq}, it is then straightforward to verify the conditions of \cref{def:admissible.X}. %
\end{proof}

\subsubsection{Proof of \cref{prop:reduced.rep.obj}}

\begin{proof}[Proof of \cref{prop:reduced.rep.obj}]
    We first note a trivial case of~\eqref{eq:dual.predictable.proj.properties}: 
    if the process $Z$ in  \cref{thm:dual.pred.proj} is itself $\mathbb{G}$-predictable, clearly $Z^{\rm p^\mathbb{G}}=Z$, and then~\eqref{eq:dual.predictable.proj.properties} includes the identity
    \[\mathbb{E}\left[\int_0^\infty\xi_s dZ_s\right]=\mathbb{E}\left[\int_0^\infty {}^{\rm p^{\mathbb{G}}}\xi_s dZ_s\right].\]
    
    Consider \eqref{eqn:obj.schied.rep} and more specifically the terms that depend on $\boldsymbol{X}^{-i}$. Denoting those by $J^{-i}(X^i;\boldsymbol{X}^{-i})$, applying \cref{thm:dual.pred.proj}, and using the $\mathbb{F}^i$-predictability of $X^i$ as just noted,
\begin{align*}
    J^{-i}(X^i;\boldsymbol{X}^{-i})
    &:=\mathbb{E}\Bigg[\int_0^T\!\!\int_0^{t-} G(t-s) \sum_{j\not=i}dX_s^{j} dX_t^{i} + \frac{G(0)}{2}\sum_{j\not=i}\sum_{t\in[0,T]} \Delta X_t^{j}\Delta X_t ^{i}\Bigg]\\
    &=\mathbb{E}\Bigg[\int_0^T\leftindex^{\;{\rm p}^i\!\!\!}{\left(\int_0^{t-}G(t-s)\sum_{j\not=i}dX_s^{j}\right)}dX_t^{i} + \frac{G(0)}{2}\sum_{j\not=i}\sum_{t\in[0,T]} \leftindex^{{\rm p}^i}(\Delta X_t^{j})\Delta X_t ^{i}\Bigg].
\end{align*}
 (To apply \cref{thm:dual.pred.proj}, note that the sum is a special case of an integral, $\sum_{t\in[0,T]} \Delta X_t^{j}\Delta X_t ^{i}= \int_0^T \Delta X_t^{j} dX_t ^{i}$.) 
Applying \cref{prop:pred.proj.private.info} we then get 
\begin{align}\mathbb{E}&\Bigg[\int_0^T\leftindex^{\;{\rm p}^i\!\!\!}{\left(\int_0^{t-} G(t-s) \sum_{j\not=i}dX_s^{j}\right)}dX_t^{i} + \frac{G(0)}{2}\sum_{j\not=i}\sum_{t\in[0,T]} \leftindex^{{\rm p}^i}(\Delta X_t^{j})\Delta X_t ^{i}\Bigg] \nonumber\\
&=\mathbb{E}\Bigg[\int_0^T\leftindex^{\;{\rm p}\!\!\!}{\left(\int_0^{t-} G(t-s)\sum_{j\not=i}dX_s^{j}\right)}dX_t^{i} + \frac{G(0)}{2}\sum_{j\not=i}\sum_{t\in[0,T]} \leftindex^{{\rm p}}(\Delta X_t^{j})\Delta X_t ^{i}\Bigg] \nonumber\\
&=\mathbb{E}\Bigg[\int_0^T\leftindex^{\;{\rm p}\!\!\!}{\left(\int_0^{t-} G(t-s) \sum_{j\not=i}dX_s^{j}\right)}dX_t^{i} + \frac{G(0)}{2}\sum_{j\not=i}\sum_{t\in[0,T]} \Delta ((X_t^{j})^{\rm p})\Delta X_t ^{i}\Bigg],\label{eqn:intermediate.eqn.pred.proj.obj}
\end{align}
where the last equality used the general fact that $\leftindex^{{\rm p}}(\Delta A)=\Delta (A^{\rm p})$ for any c\`adl\`ag adapted process $A$ with integrable variation (see \cite[3.21, p.\,33]{JacodShiryaev.03}). Denote the first integrand in~\eqref{eqn:intermediate.eqn.pred.proj.obj} by
\[I_t^{-i}:=\int_0^{t-} G(t-s) \sum_{j\not=i}dX_s^{j}. \]%
From \cref{prop:pred.is.exp} we know that ${}^{\rm p}I_t^{-i} = \tilde{\mathbb{E}}[I_t^{-i}]$ and we show below that 
\begin{align}\label{eq:proof.reduced.rep.obj}
    \tilde{\mathbb{E}}[I_t^{-i}] &=
    \int_0^{t-} G(t-s)\sum_{j\not=i} d\tilde{\mathbb{E}}[X^j_s].
\end{align}
Next, we substitute these identities into the first term of \eqref{eqn:intermediate.eqn.pred.proj.obj}.
Using also that $(X^j)^{\rm p} = \tilde{\mathbb{E}}[X^j]$ by \cref{prop:pred.is.exp} to rewrite the second term of of \eqref{eqn:intermediate.eqn.pred.proj.obj}, we get
\begin{align*}
    J^{-i}(X^i;\boldsymbol{X}^{-i})
    &=\mathbb{E}\Bigg[\int_0^T\!\!\int_0^{t-} G(t-s) \sum_{j\not=i}d\tilde{\mathbb{E}}[X_s^{j}] dX_t^{i} + \frac{G(0)}{2}\sum_{j\not=i}\sum_{t\in[0,T]} \Delta \tilde{\mathbb{E}}[X^j_s]\Delta X_t ^{i}\Bigg].
\end{align*}
Comparing with \eqref{eqn:obj.schied.rep} we see that the right hand side equals $J^{-i}(X^i;\tilde{\mathbb{E}}[\boldsymbol{X}^{-i}])$, thus establishing that $J(X^i;\boldsymbol{X}^{-i})=J(X^i;\tilde{\mathbb{E}}[\boldsymbol{X}^{-i}])$.

It remains to show~\eqref{eq:proof.reduced.rep.obj}. This is clear for $t=0$ so we may focus on $t>0$. Recall that~$G$ is continuous on $[0,T]$. Suppose for a moment that $G$ is even continuously differentiable. Then integrating by parts gives
\begin{align*}
    \int_0^r G(r-s)\sum_{j\not=i} dX^{j}_s
    = G(0) \sum_{j\not=i} X^j_r - G(r) \sum_{j\not=i} X^j_{0-} +\int_0^r G'(r-s) \sum_{j\not=i} X^j_s  ds.
\end{align*}
Taking expectations, applying Fubini's theorem, and reversing the integration by parts yields
\begin{align}
    \tilde{\mathbb{E}}\left[\int_0^r G(r-s)\sum_{j\not=i} dX^{j}_s\right] &=G(0) \sum_{j\not=i} \tilde{\mathbb{E}}[X^j_r] - G(r) \sum_{j\not=i} \tilde{\mathbb{E}}[X^j_{0-}] +\int_0^r G'(r-s) \sum_{j\not=i}\tilde{\mathbb{E}}[X^j_s]  ds \nonumber\\
    &=\int_0^r G(r-s)\sum_{j\not=i} d\tilde{\mathbb{E}}[X^j_s]. \label{eq:for.proof.reduced.rep.obj}
\end{align}
In fact, these steps hold for any continuously differentiable function~$G$. The general continuous kernel $G$ can be approximated uniformly on~$[0,T]$ with continuously differentiable functions, hence it follows that~\eqref{eq:for.proof.reduced.rep.obj} also holds for the general kernel~$G$. The claim~\eqref{eq:proof.reduced.rep.obj} now follows by taking the limit $r\uparrow t$ and using that~$G$ is continuous and bounded.
\end{proof}

\subsubsection{Counterexample regarding \cref{as:C.no.randomized}}\label{app:counterexample}

We give a counterexample showing that Assumption~4.3 is \emph{not} automatic from convexity of $C$ on our set of admissible strategies.

Fix $T>0$ and consider the deterministic setup $\Omega^{0}=\{\omega^{0}\}$ with
$\mathcal F^{0}=\{\emptyset,\Omega^{0}\}$, $\mathbb P^{0}(\Omega^{0})=1$, and the trivial filtration
$\mathbb F^{0}=(\mathcal F^{0}_{t})_{t\ge0}$ with $\mathcal F^{0}_{t}=\mathcal F^{0}$ for all $t$.
Suppose there is a single player ($N=1$) and let the randomization space be $\tilde\Omega=\mathbb N=\{1,2,\dots\}$ with $\tilde{\mathcal F}=2^{\mathbb N}$ and a probability
$\tilde{\mathbb P}$ assigning strictly positive mass to each $n\in\mathbb N$; e.g.
$\tilde{\mathbb P}(\{n\})=2^{-n}$ for $n\ge1$.
Set $\Omega:=\Omega^{0}\times\tilde\Omega$, $\mathcal F:=\mathcal F^{0}\otimes\tilde{\mathcal F}$ and
$\mathbb P:=\mathbb P^{0}\otimes\tilde{\mathbb P}$, and let the market filtration $\mathbb F^{M}$ be the usual
augmentation of $(\mathcal F^{0}_{t}\otimes\sigma(\{\emptyset,\tilde\Omega\}))_{t\ge0}$, so that $\mathbb F^{M}$ is
trivial.

Choose a strictly increasing sequence of times $(t_n)_{n\ge1}\subset(0,T)$ with $t_n\uparrow T$.
Define a process $X=(X_t)_{t\ge0}$ on $\Omega$ by
\[
X_t(\omega^{0},n):=\mathbf 1_{[t_n,\infty)}(t),\qquad t\ge0,\ n\in\mathbb N.
\]
Then $X$ is c\`adl\`ag, of bounded variation, and constant for $t\ge T$. Moreover, since $\mathbb F^{1}$ is the usual
augmentation of $(\mathcal F^{0}_{t}\otimes\sigma(\tilde\pi_{1}))_{t\ge0}$, the process $X$ is trivially
$\mathbb F^{1}$-predictable and therefore admissible for initial data $x^1=0$.

Its marginal expectation is the deterministic (hence $\mathbb F^{M}$-predictable) c\`adl\`ag function
\[
\tilde{\mathbb E}[X_t]
=\sum_{n\ge1}\tilde{\mathbb P}(\{n\})\,\mathbf 1_{[t_n,\infty)}(t),\qquad t\geq0,
\]
which has jumps of size $\tilde{\mathbb P}(\{n\})>0$ at each time $t_n$ and therefore has \emph{infinitely many}
jumps on $[0,T]$.

Now let $\mathcal X$ denote the set of real-valued c\`adl\`ag functions $z:[0,\infty)\to\mathbb R$ of finite variation
such that $z_{0-}=0$ and $z_t=z_T$ for all $t\ge T$ (i.e., the path space of admissible strategies with initial inventory $0$). Define the convex set
\[
K:=\Big\{z\in\mathcal X:\ z \text{ has only finitely many jumps on }[0,T] \Big\}.
\]
Then $K$ is convex: if $y,z\in K$ and $\alpha\in[0,1]$, the jump times of $\alpha y+(1-\alpha)z$ are contained in the
finite union of the jump time sets of $y$ and $z$, hence $\alpha y+(1-\alpha)z\in K$.
Define the extended-valued cost functional $C:\mathcal X\to[0,\infty]$ as the indicator of $K$,
\[
C(z):=\iota_{K}(z):=
\begin{cases}
0, & z\in K,\\
\infty, & z\notin K.
\end{cases}
\]
This $C$ is convex (as the indicator of a convex set).

By construction, $X(\cdot,n)\in K$ for every $n$ (it has exactly one jump), hence $C(X)=0$ $\mathbb P$-a.s. and thus
$\mathbb E[C(X)]=0$.
On the other hand, $\tilde{\mathbb E}[X]\notin K$ because it has infinitely many jumps, so
$C(\tilde{\mathbb E}[X])=\infty$ and $\mathbb E[C(\tilde{\mathbb E}[X])]=\infty$.
Therefore,
\[
\mathbb E[C(X)] < \mathbb E\big[C(\tilde{\mathbb E}[X])\big],
\]
showing that Assumption~4.3 does not follow from convexity of $C$ alone.

\subsubsection{Proof of \cref{le:C.no.randomized.holds.for.quadratic}}

Next, we show that the quadratic costs of \cref{ex:costsC} satisfy the Jensen-type inequality of \cref{as:C.no.randomized}.
    
\begin{proof}[Proof of \cref{le:C.no.randomized.holds.for.quadratic}]
    If $\mathbb{E}\left[C(X)\right]=\infty$, the conclusion is immediate. 
    We therefore fix an admissible strategy $X$ for some trader $i$ with $\mathbb{E}\left[C(X)\right]<\infty$.
    Observe that by Fubini's theorem, %
    \begin{equation}\mathbb{E}\left[C(X)\right]
    = \int_{\Omega^{\noll}}\tilde{\mathbb{E}}[C(X)]d\mathbb{P}^{\noll}=\int_{\Omega^{\noll}}\int_{\tilde{\Omega}}\tilde{\mathbb{E}}[C(X)]d\tilde{\mathbb{P}}d\mathbb{P}^{\noll}=\mathbb{E}\left[\tilde{\mathbb{E}}[C(X)]\right].\label{eqn:C.init.equality}
    \end{equation}
    Invoking Fubini's theorem once again with a pointwise application of Jensen's inequality on the quadratic function $z\mapsto z^2$ for $z\in\mathbb{R}$,
\begin{align}\mathbb{E}\left[\tilde{\mathbb{E}}[C(X)]\right]&=
		\mathbb{E}\left[\tilde{\mathbb{E}}\left[\frac{\varepsilon}{2}\int_0^T (\dot{X}_t)^2dt
		+\frac{1}{2}\sum_{t\in[0,T]} \vartheta_t(\Delta X_t)^2
		+\frac{\varphi}{2} (X_T)^2\right]\right]\nonumber\\
        &\geq \mathbb{E}\left[\frac{\varepsilon}{2}\int_0^T \left(\tilde{\mathbb{E}}\left[\dot{X}_t\right]\right)^2dt
		+\tilde{\mathbb{E}}\left[\frac{1}{2}\sum_{t\in[0,T]} \vartheta_t\left(\Delta X_t\right)^2\right]
		+\frac{\varphi}{2} \left(\tilde{\mathbb{E}}\left[X_T\right]\right)^2\right].%
        \label{eqn:C.pt.1}
\end{align}
With a view towards the first term, we next argue that if $\varepsilon>0$, then $\frac{d}{dt}\tilde{\mathbb{E}}\left[X_t\right]=\tilde{\mathbb{E}}[\dot{X}_t]$ $dt\otimes d\mathbb{P}$ almost surely. 
Indeed, if $\varepsilon>0$, then by the definition in \cref{ex:costsC} and $\mathbb{E}\left[C(X)\right]<\infty$ we know that $X$ is absolutely continuous. 
Thus, given the initial value $\tilde{\mathbb{E}}[X_{0}]=X_{0} = x^i\in\mathbb{R}$,
\[\tilde{\mathbb{E}}[X_t] = x^i+ \tilde{\mathbb{E}}\left[\int_0^t \dot{X}_sds\right]. \]
Since $X$ is admissible, it has $\mathbb{P}$-essentially bounded variation and in particular
$\int_0^t |\dot{X}_s|ds$ is uniformly bounded. Thus, we may apply Fubini's theorem to obtain
$\tilde{\mathbb{E}}[X_t] = x^i+\int_0^t \tilde{\mathbb{E}} [\dot{X}_s]ds$. In other words, $\tilde{\mathbb{E}}[X_t]$ is absolutely continuous with derivative $\frac{d}{dt}\tilde{\mathbb{E}}\left[X_t\right]=\tilde{\mathbb{E}}\left[\dot{X}_t\right]$ $dt\otimes d\mathbb{P}$ almost surely, and we have
\begin{equation}\label{eqn:C.pt.1.firstTerm}
    \mathbb{E}\left[\tilde{\mathbb{E}}\left[\frac{\varepsilon}{2}\int_0^T (\dot{X}_t)^2dt \right]\right]
    \geq\mathbb{E}\left[\frac{\varepsilon}{2}\int_0^T \left(\frac{d}{dt}\tilde{\mathbb{E}}\left[X_t\right]\right)^2dt \right].
\end{equation}

It remains to treat the second term in~\eqref{eqn:C.pt.1}. By \cref{thm:dual.pred.proj} and the identity $X^{\rm p}=\tilde{\mathbb{E}}[X]$ from \cref{prop:pred.is.exp},
\begin{align*}0&\leq \mathbb{E}\left[
		\frac{1}{2}\sum_{t\in[0,T]} \vartheta_t\left(\Delta X_t^{i}-\Delta \tilde{\mathbb{E}}[X_t^{i}]\right)^2\right]\\
        &=\mathbb{E}\left[\frac{1}{2}\sum_{t\in[0,T]} \vartheta_t\left((\Delta X_t)^2-2\Delta X_t\Delta \tilde{\mathbb{E}}[X_t]+(\Delta \tilde{\mathbb{E}}[X_t])^2\right)\right]\\
        &=\mathbb{E}\left[\frac{1}{2}\sum_{t\in[0,T]} \vartheta_t(\Delta X_t)^2-\frac{1}{2}\sum_{t\in[0,T]} \vartheta_t(\Delta \tilde{\mathbb{E}}[X_t])^2\right] \\
        &\quad +\mathbb{E}\left[\sum_{t\in[0,T]} \vartheta_t(\Delta \tilde{\mathbb{E}}[X_t])^2-\sum_{t\in[0,T]} \vartheta_t\Delta X_t\Delta \tilde{\mathbb{E}}[X_t]\right]\\
        &=\mathbb{E}\left[\frac{1}{2}\sum_{t\in[0,T]} \vartheta_t(\Delta X_t)^2-\frac{1}{2}\sum_{t\in[0,T]} \vartheta_t(\Delta \tilde{\mathbb{E}}[X_t])^2\right] \\
        & \quad +\mathbb{E}\left[\sum_{t\in[0,T]} \vartheta_t(\Delta \tilde{\mathbb{E}}[X_t])^2-\sum_{t\in[0,T]} \vartheta_t\Delta X_t^{\rm p}\Delta \tilde{\mathbb{E}}[X_t]\right]\\
        &=\mathbb{E}\left[\frac{1}{2}\sum_{t\in[0,T]} \vartheta_t(\Delta X_t)^2-\frac{1}{2}\sum_{t\in[0,T]} \vartheta_t(\Delta \tilde{\mathbb{E}}[X_t])^2\right].
\end{align*}
By rearranging we conclude that
\begin{equation}\label{eqn:C.pt.2}\mathbb{E}\left[\tilde{\mathbb{E}}\left[\frac{1}{2}\sum_{t\in[0,T]} \vartheta_t(\Delta X_t)^2\right]\right]=\mathbb{E}\left[\frac{1}{2}\sum_{t\in[0,T]} \vartheta_t(\Delta X_t)^2\right] \geq \mathbb{E}\left[\frac{1}{2}\sum_{t\in[0,T]} \vartheta_t(\Delta \tilde{\mathbb{E}}[X_t])^2\right].
\end{equation}
Combining \cref{eqn:C.init.equality,eqn:C.pt.1.firstTerm,eqn:C.pt.1,eqn:C.pt.2} completes the proof.
\end{proof}

\subsubsection{Proof of \cref{thm:non.existence.random.eq}}

We first prove the following technical lemma.

\begin{lemma}\label{lem:jensen.ineq.precursor}
    For any %
    kernel $G$ satisfying \cref{as:kernel} and any admissible strategy~$X$,
    \begin{align*}
    \mathbb{E}&\left[\frac{1}{2}\int_0^\infty\!\!\int_0^\infty G(|t-s|) d(X_s-\tilde{\mathbb{E}}[X_s]) d(X_t-\tilde{\mathbb{E}}[X_t]) \right]\\
    &\quad \quad \quad \quad \quad =\mathbb{E}\left[\frac{1}{2}\int_0^\infty\!\!\int_0^\infty G(|t-s|) dX_sdX_t -\frac{1}{2}\int_0^\infty\!\!\int_0^\infty G( |t-s|)d\tilde{\mathbb{E}}[X_s] d\tilde{\mathbb{E}}[X_t]\right].
\end{align*}
\end{lemma}

\begin{proof}
    Since $X$ and $\tilde{\mathbb{E}}[X]$ are constant after $T$ it suffices to check the equality for upper limit of integration $T$. We have
    
    \begin{align}
    \mathbb{E}\bigg[\frac{1}{2}\int_0^T\!\!\int_0^T& G(|t-s|) d(X_s^{i}-\tilde{\mathbb{E}}[X_s^i]) d(X_t^{i}-\tilde{\mathbb{E}}[X_t^i]) \bigg]\nonumber\\
    &=\mathbb{E}\bigg[\frac{1}{2}\int_0^T\!\!\int_0^T G(|t-s|) dX_s^{i}dX_t^{i} +\frac{1}{2}\int_0^T\!\!\int_0^T G(|t-s|)d\tilde{\mathbb{E}}[X_s^i] d\tilde{\mathbb{E}}[X_t^i] \nonumber\\
    &\quad \quad \quad -\int_0^T\!\!\int_0^T G(|t-s|)d\tilde{\mathbb{E}}[X_s^i]dX_t^i \bigg]\nonumber\\
    &=\mathbb{E}\bigg[\frac{1}{2}\int_0^T\!\!\int_0^T G(|t-s|) dX_s^{i}dX_t^{i} -\frac{1}{2}\int_0^T\!\!\int_0^T G(|t-s|)d\tilde{\mathbb{E}}[X_s^i] d\tilde{\mathbb{E}}[X_t^i] \nonumber\\
    &\quad \quad \quad +\int_0^T\!\!\int_0^T G(|t-s|)d\tilde{\mathbb{E}}[X_s^i] d\tilde{\mathbb{E}}[X_t^i] -\int_0^T\!\!\int_0^T G(|t-s|)d\tilde{\mathbb{E}}[X_s^i]dX_t^i \bigg].\label{eqn:intermediate.calc.non.exist.1}
\end{align}
Next, note that for arbitrary Lebesgue--Stieltjes integrators $M,N$ we have
\begin{align*}
    \int_0^T\!\!\int_0^{t-} &G(t-s) dM_sdN_t \\
    & = \frac{1}{2}\int_0^T\!\!\int_0^{t-} G(t-s)dM_sdN_t+ \frac{1}{2}\int_0^T\!\!\int_0^{t-} G(t-s) dM_sdN_t\\
    &=\frac{1}{2}\int_0^T\!\!\int_0^{t-} G(t-s) dM_sdN_t+ \frac{1}{2}\int_0^T\!\!\int_s^{T} G(t-s) dN_tdM_s-\frac{G(0)}{2}\sum_{t\in[0,T]}\Delta M _t\Delta N_t\\
    &=\frac{1}{2}\int_0^T\!\!\int_0^{t-} G(t-s) dM_sdN_t+ \frac{1}{2}\int_0^T\!\!\int_t^T G(s-t) dM_sdN_t -\frac{G(0)}{2}\sum_{t\in[0,T]}\Delta M _t\Delta N_t \\
    &\quad \quad +\frac{1}{2}\int_0^T\!\!\int_t^{T} G(s-t) dN_sdM_t - \frac{1}{2}\int_0^T\!\!\int_t^T G(s-t)dM_sdN_t \\
    &=\frac{1}{2}\int_0^T\!\!\int_0^{T} G(|t-s|) dM_sdN_t-\frac{G(0)}{2}\sum_{t\in[0,T]}\Delta M _t\Delta N_t\\
    &\quad \quad +\frac{1}{2}\int_0^T\!\!\int_0^{t-} G(t-s) dM_sdN_t - \frac{1}{2}\int_0^T\!\!\int_0^{t-} G(t-s) dN_sdM_t.
\end{align*}
Rearranging, this leads to the identity
\begin{align} 
\frac{1}{2}\int_0^T\!\!\int_0^{T} G(|t-s|) dM_sdN_t&= \int_0^T\!\!\int_0^{t-} G(t-s) dM_sdN_t+\frac{G(0)}{2}\sum_{t\in[0,T]}\Delta M _t\Delta N_t \label{eqn:int.identity}\\
&\quad +\frac{1}{2}\int_0^T\!\!\int_0^{t-} G(t-s) dN_sdM_t -\frac{1}{2}\int_0^T\!\!\int_0^{t-} G(t-s)dM_sdN_t.\nonumber
\end{align}
In particular, if $M\equiv N$,
\begin{align} 
\frac{1}{2}\int_0^T\!\!\int_0^{T} G(|t-s|) dM_sdM_t= \int_0^T\!\!\int_0^{t-} G(t-s) dM_sdM_t+\frac{G(0)}{2}\sum_{t\in[0,T]}\Delta M _t\Delta M_t.\label{eqn:int.identity.2}
\end{align}
Using \eqref{eqn:int.identity} with $M_s=\tilde{\mathbb{E}}[X_s^i]$ and $N_t=X_t^i$ and then applying \cref{thm:dual.pred.proj} yields
\begin{align*}
    \mathbb{E}\Bigg[\int_0^T\!\!\int_0^T &G(|t-s|) d\tilde{\mathbb{E}}[X_s^i]dX_t^i\Bigg]\\
    &=\mathbb{E}\Bigg[2\int_0^T\!\!\int_0^{t-} G(t-s) d\tilde{\mathbb{E}}[X_s^i]dX_t^i+G(0)\sum_{t\in[0,T]}\Delta \tilde{\mathbb{E}}[X_t^i]\Delta X_t^i\\
    &\quad \quad + \int_0^T\!\!\int_0^{t-} G(t-s) dX_s^id\tilde{\mathbb{E}}[X_t^i] -\int_0^T\!\!\int_0^{t-} G(t-s) d\tilde{\mathbb{E}}[X_s^i]dX_t^i\Bigg]\\
    &=\mathbb{E}\Bigg[2\int_0^T\!\!\int_0^{t-} G(t-s) d\tilde{\mathbb{E}}[X_s^i]d(X_t^i)^{\rm p}+G(0)\sum_{t\in[0,T]}\Delta \tilde{\mathbb{E}}[X_t^i]\Delta (X_t^i)^{\rm p}\\
    &\quad \quad + \int_0^T\leftindex^{\;\rm p\!\!\!}{\left(\int_0^{t-} G(t-s) dX_s^i\right)}d\tilde{\mathbb{E}}[X_t^i] -\int_0^T\!\!\int_0^{t-} G(t-s) d\tilde{\mathbb{E}}[X_s^i]d(X_t^i)^{\rm p}\Bigg]\\
    &=\mathbb{E}\Bigg[2\int_0^T\!\!\int_0^{t-} G(t-s) d\tilde{\mathbb{E}}[X_s^i]d\tilde{\mathbb{E}}[X_t^i]+G(0)\sum_{t\in[0,T]}\left(\Delta \tilde{\mathbb{E}}[X_t^i]\right)^2\\
    &\quad \quad + \int_0^T\leftindex^{\;\rm p\!\!\!}{\left(\int_0^{t-} G(t-s) dX_s^i\right)}d\tilde{\mathbb{E}}[X_t^i] -\int_0^T\!\!\int_0^{t-} G(t-s) d\tilde{\mathbb{E}}[X_s^i]d\tilde{\mathbb{E}}[X_t^i]\Bigg].
\end{align*}
As seen in the proof of \cref{prop:reduced.rep.obj},
\[\leftindex^{\;\rm p\!\!\!}{\left(\int_0^{t-} G(t-s) dX_s^i\right)} = \int_0^{t-} G(t-s) d\tilde{\mathbb{E}}[X_s^i].\]
Substituting this, simplifying, and applying \eqref{eqn:int.identity.2} with $M_s=\tilde{\mathbb{E}}[X_s^i]$,
\begin{align}
    \mathbb{E}\Bigg[\int_0^T\!\!\int_0^T &G(|t-s|)d\tilde{\mathbb{E}}[X_s^i]dX_t^i\Bigg]\nonumber\\&=\mathbb{E}\Bigg[2\int_0^T\!\!\int_0^{t-} G(t-s) d\tilde{\mathbb{E}}[X_s^i]d\tilde{\mathbb{E}}[X_t^i]+G(0)\sum_{t\in[0,T]}\left(\Delta \tilde{\mathbb{E}}[X_t^i]\right)^2\nonumber\\
    &\quad \quad + \int_0^T\!\!\int_0^{t-} G(t-s) d\tilde{\mathbb{E}}[X^i_s]d\tilde{\mathbb{E}}[X_t^i] -\int_0^T\!\!\int_0^{t-} G(t-s)d\tilde{\mathbb{E}}[X_s^i]d\tilde{\mathbb{E}}[X_t^i]\Bigg]\nonumber\\
    &=\mathbb{E}\Bigg[2\int_0^T\!\!\int_0^{t-} G(t-s) d\tilde{\mathbb{E}}[X_s^i]d\tilde{\mathbb{E}}[X_t^i]+G(0)\sum_{t\in[0,T]}\left(\Delta \tilde{\mathbb{E}}[X_t^i]\right)^2\Bigg]\nonumber\\
    &=\mathbb{E}\left[\int_0^T\!\!\int_0^T G(|t-s|)d\tilde{\mathbb{E}}[X_s^i]d\tilde{\mathbb{E}}[X_t^i]\right].\label{eqn:intermediate.calc.non.exist.2}
\end{align}
Substituting \eqref{eqn:intermediate.calc.non.exist.2} into \eqref{eqn:intermediate.calc.non.exist.1} and simplifying completes the proof.
\end{proof}

Next, we prove the main result.

\begin{proof}[Proof of \cref{thm:non.existence.random.eq}]
Using \cref{prop:reduced.rep.obj} and~\eqref{eqn:obj.schied.rep} and the definition of the dual predictable projection in \cref{thm:dual.pred.proj},
\begin{align}
J(X^i;\boldsymbol{X}^{-i}) 
&=J(X^i;\tilde{\mathbb{E}}[\boldsymbol{X}^{-i}]) \nonumber \\
&=\mathbb{E}\Bigg[\frac{1}{2}\int_0^T\!\!\int_0^T G(|t-s|) dX_s^{i} dX_t^i + \int_0^T\!\!\int_0^{t-} G(t-s) \sum_{j\not=i}d\tilde{\mathbb{E}}[X_s^{j}] d(X_t^{i})^{\rm p} \nonumber \\
        &\quad \quad \quad \quad + \frac{G(0)}{2}\sum_{j\not=i}\sum_{t\in[0,T]} \Delta \tilde{\mathbb{E}}[X_t^{j}]\Delta(X_t^{i})^{\rm p} +C(X^{i})\Bigg]\nonumber\\
        &=\mathbb{E}\Bigg[\frac{1}{2}\int_0^T\!\!\int_0^T G(|t-s|) dX_s^{i} dX^i_t + \int_0^T\!\!\int_0^{t-} G(t-s)\sum_{j\not=i}d\tilde{\mathbb{E}}[X_s^{j}] d\tilde{\mathbb{E}}[X^i_t] \nonumber \\
        &\quad \quad \quad \quad + \frac{G(0)}{2}\sum_{j\not=i}\sum_{t\in[0,T]} \Delta \tilde{\mathbb{E}}[X_t^{j}]\Delta \tilde{\mathbb{E}}[X^i_t] +C(X^{i})\Bigg],\label{eqn:obj.rep.non.exist.pf}
    \end{align}
where the last equality used, similarly as in the proof of \cref{prop:reduced.rep.obj}, that $\leftindex^{{\rm p}}X^i=(X^i)^{\rm p} = \tilde{\mathbb{E}}[X^i]$ by \cref{prop:pred.is.exp} and the rule $\leftindex^{{\rm p}}(\Delta A)=\Delta (A^{\rm p})$. 

By our assumption, $X^i$ cannot be equal to $\tilde{\mathbb{E}}[X^i]$ up to evanescence. As $x^i=X^i_{0-} = \tilde{\mathbb{E}}[X^i_{0-}]$, this implies that the random measure $d(X_\cdot^{i}-\tilde{\mathbb{E}}[X_\cdot^i])$ is non-zero with positive probability. Recalling that~$G$ is strict positive definite and invoking \cref{lem:jensen.ineq.precursor}, it follows that

\begin{align*}
    0<\mathbb{E}&\left[\frac{1}{2}\int_0^T\!\!\int_0^T G(|t-s|) d(X_s-\tilde{\mathbb{E}}[X_s]) d(X_t-\tilde{\mathbb{E}}[X_t]) \right]\\
    &\quad \quad \quad \quad \quad =\mathbb{E}\left[\frac{1}{2}\int_0^T\!\!\int_0^T G(|t-s|) dX_sdX_t -\frac{1}{2}\int_0^T\!\!\int_0^T G( |t-s|)d\tilde{\mathbb{E}}[X_s] d\tilde{\mathbb{E}}[X_t]\right]
\end{align*}
which after rearranging states that
\[\mathbb{E}\left[\int_0^T\!\!\int_0^T G(|t-s|) dX_s^{i}dX_t^{i}\right]>\mathbb{E}\left[\int_0^T\!\!\int_0^T G(|t-s|)d\tilde{\mathbb{E}}[X_s^i] d\tilde{\mathbb{E}}[X_t^i]\right].\]
Combining this with \eqref{eqn:obj.rep.non.exist.pf}, we obtain
\begin{align*}J(X^i;\boldsymbol{X}^{-i}) &>\mathbb{E}\Bigg[\frac{1}{2}\int_0^T\!\!\int_0^T G( |t-s|) d\tilde{\mathbb{E}}[X_s^i] d\tilde{\mathbb{E}}[X^i_t] + \int_0^T\!\!\int_0^{t-} G(t-s) \sum_{j\not=i}d\tilde{\mathbb{E}}[X_s^{j}] d\tilde{\mathbb{E}}[X^i_t] \nonumber \\
        &\quad \quad \quad \quad + \frac{G(0)}{2}\sum_{j\not=i}\sum_{t\in[0,T]} \Delta \tilde{\mathbb{E}}[X_t^{j}]\Delta \tilde{\mathbb{E}}[X^i_t] + C(X^{i})\Bigg ].
    \end{align*}
By \cref{as:C.no.randomized} we have $\mathbb{E}[C(X^i)]\geq \mathbb{E}[C(\tilde{\mathbb{E}}[X^i])]$. Thus we may replace $C(X^i)$ by $C(\tilde{\mathbb{E}}[X^i])$ in the above inequality, and then the right hand side is precisely $J(\tilde{\mathbb{E}}[X^i];\tilde{\mathbb{E}}[\boldsymbol{X}^{-i}])$.
Applying \cref{prop:reduced.rep.obj} once more yields the claim:
\[J(X^i;\boldsymbol{X}^{-i})>J(\tilde{\mathbb{E}}[X^i];\tilde{\mathbb{E}}[\boldsymbol{X}^{-i}])=J(\tilde{\mathbb{E}}[X^i];\boldsymbol{X}^{-i}). \qedhere\]
\end{proof}

\subsection{Proofs for \cref{se:uniqueness} (Uniqueness)}\label{se:proofs.uniqueness}

The first step towards uniqueness of Nash equilibria is to ensure that the traders' objectives are convex in their controls. (In fact, the proof of our main result below only uses the assertion of convexity, not the strict convexity provided in the lemma.)

\begin{lemma}\label{lem:strict.convex}
    For any admissible $\boldsymbol{X}^{-i}$,  the objective $J(\,\cdot\,; \boldsymbol{X}^{-i})$ is strictly\footnote{An extended real-valued convex function $F$ is called strictly convex if it is strictly convex on its domain
$\mathrm{dom}(F)=\{h:F(h)<\infty\}$.} convex in its first argument.
\end{lemma}

The proof is analogous to the proof in \cite[Lemma 4.7]{SchiedStrehleZhang.17} for the exponential kernel and a quadratic cost $C$. We include it for completeness and to clarify that it relies only on the convexity of the cost~$C$ and the strict positive definiteness of the kernel~$G$.

\begin{proof}
    Let $Y^0$ and $Y^1$ be distinct admissible strategies for trader $i$ with $\mathbb{E}[C(Y^0)]<\infty$ and $\mathbb{E}[C(Y^1)]<\infty$. Fix any admissible strategy profile $\boldsymbol{X}^{-i}$ of the other traders. For any $\alpha\in(0,1)$ define
    \[Y^\alpha:=\alpha Y^1+(1-\alpha) Y^0\]
    so that by the convexity of $C$,
    \begin{align*}
    J(Y^\alpha;\boldsymbol{X}^{-i})&\leq \mathbb{E}\Bigg[\frac{1}{2}\int_0^T\!\!\int_0^T G(|t-s|)dY^\alpha_s dY^\alpha_t + \int_0^T\!\!\int_0^{t-} G(t-s) \sum_{j\not=i}dX_s^{j} dY^\alpha_t \nonumber \\
    &\quad \quad \quad \quad + \frac{G(0)}{2}\sum_{j\not=i}\sum_{t\in[0,T]} \Delta X_t^{j}\Delta Y^\alpha_t +\alpha C(Y^1) + (1-\alpha) C(Y^0)\Bigg].
    \end{align*}
    Observe that the middle two terms are linear in $\alpha$, so to show strict convexity of $J(\cdot;\boldsymbol{X}^{-i})$ it suffices to show
    \begin{align}H(Y^\alpha):=\mathbb{E}\left[\int_0^T\!\!\int_0^T G(|t-s|)dY^\alpha_s dY^\alpha_t\right]&<\alpha\mathbb{E}\left[\int_0^T\!\!\int_0^T G(|t-s|)dY^1_s dY^1_t\right]\label{eqn:suff.cond.convexity}\\
    &\quad \quad +(1-\alpha)\mathbb{E}\left[\int_0^T\!\!\int_0^T G(|t-s|)dY^0_s dY^0_t\right].\nonumber
    \end{align}
    Expanding the left hand side of \eqref{eqn:suff.cond.convexity},
    \begin{align}
        H(Y^\alpha)&=\alpha^2\mathbb{E}\left[\int_0^T\!\!\int_0^T G(|t-s|)dY^1_t dY^1_t\right]+(1-\alpha)^2\mathbb{E}\left[\int_0^T\!\!\int_0^T G(|t-s|)dY^0_s dY^0_t\right]\nonumber\\
        &\quad \quad +2\alpha(1-\alpha) \mathbb{E}\left[\int_0^T\!\!\int_0^T G(|t-s|)dY^1_s dY^0_t\right].\label{eqn:pf.convexity.1}
    \end{align}
    At the same time, by the strict positive definiteness of $G$,
    \begin{align}
    0&<\alpha(1-\alpha)\mathbb{E}\left[\int_0^T\!\!\int_0^T G(|t-s|)d(Y^1_s-Y^0_s) d(Y^1_t-Y^0_t)\right]\nonumber\\
    &=\alpha(1-\alpha)\mathbb{E}\left[\int_0^T\!\!\int_0^T G(|t-s|)dY^1_s dY^1_t\right]+\alpha(1-\alpha)\mathbb{E}\left[\int_0^T\!\!\int_0^T G(|t-s|)dY^0_s dY^0_t\right]\nonumber\\
    &\quad \quad -2\alpha(1-\alpha)\mathbb{E}\left[\int_0^T\!\!\int_0^T G(|t-s|)dY^1_s dY^0_t\right]. \label{eqn:pf.convexity.2}
    \end{align}
    Combining \eqref{eqn:pf.convexity.1} and \eqref{eqn:pf.convexity.2} we obtain
    \begin{align*}
        H(Y^\alpha) &< (\alpha^2+\alpha(1-\alpha))\mathbb{E}\left[\int_0^T\!\!\int_0^T G(|t-s|)dY^1_t dY^1_t\right]\\
        &\quad \quad +((1-\alpha)^2+\alpha(1-\alpha))\mathbb{E}\left[\int_0^T\!\!\int_0^T G(|t-s|)dY^0_s dY^0_t\right].
    \end{align*}
    Simplifying the coefficients recovers \eqref{eqn:suff.cond.convexity} and completes the proof.
\end{proof}

\begin{proof}[Proof of \cref{thm:nash.eq.unique}]
Let $\boldsymbol{X}^0$ and $\boldsymbol{X}^1$ be distinct equilibria. Define
    \[\boldsymbol{X}^\alpha = \alpha \boldsymbol{X}^1 +(1-\alpha)\boldsymbol{X}^0, \ \ \ \alpha\in(0,1), \]
and
\begin{equation*}V(\alpha) = \sum_{i=1}^N\left(J(X^{\alpha,i};\boldsymbol{X}^{0,-i})+J(X^{1-\alpha,i};\boldsymbol{X}^{1,-i})\right).\end{equation*}
It follows from \cref{lem:strict.convex} is that $V$ is strictly convex. Moreover, the Nash equilibrium property implies that
\begin{equation}\label{eqn:min.V}V(\alpha)\geq \sum_{i=1}^N \left(J(X^{0,i};\boldsymbol{X}^{0,-i})+J(X^{1,i};\boldsymbol{X}^{1,-i}) \right)=V(0).
\end{equation}

Next, we will show that the right derivative at zero satisfies $\dot{V}(0+)<0$. This will be  a contradiction to \eqref{eqn:min.V} as it implies that $0$ is not in the subdifferential of $V$ at $\alpha=0$ which, in turn, implies that $V(0)$ is not a minimum.

We treat separately the four terms in the expression~\eqref{eqn:obj.schied.rep} for $J$ and their corresponding terms in~$V$. Let 
\begin{align*}F_1^i(\alpha) &= \mathbb{E}\left[\frac{1}{2}\int_0^T\!\!\int_0^T G(|t-s|) dX_s^{\alpha,i} dX_t^{\alpha, i}+\frac{1}{2}\int_0^T\!\!\int_0^T G(|t-s|) dX_s^{1-\alpha,i} dX_t^{1-\alpha, i}\right].
\end{align*}
Substituting the explicit form of $\boldsymbol{X}^\alpha,\boldsymbol{X}^{1-\alpha}$ and using the linearity of integration with respect to scalar multiplication, we see that $F_1^i(\alpha)$ is quadratic in~$\alpha$. Differentiating at $\alpha=0$,
\begin{align*}
    \dot{F}^i_1(0) &=-\mathbb{E}\left[\int_0^T\!\!\int_0^T G(|t-s|) d(X_s^{1,i} -X_s^{0, i})d(X_t^{1,i} -X_t^{0, i})\right]
\end{align*}
and thus
\begin{equation} \label{eqn:J1.sum} \sum_{i=1}^N\dot{F}^i_1(0)=-\mathbb{E}\left[\sum_{i=1}^N\int_0^T\!\!\int_0^T G( |t-s|) d(X_s^{1,i} -X_s^{0, i})d(X_t^{1,i} -X_t^{0, i})\right].
\end{equation}
Similarly, let
\begin{align*}
    F_2^i(\alpha) &=  \mathbb{E}\left[\int_0^T\!\!\int_0^{t-} G(t-s)\sum_{j\not=i}dX_s^{0,j} dX_t^{\alpha, i}+ \int_0^T\!\!\int_0^{t-} G(t-s) \sum_{j\not=i}dX_s^{1,j} dX_t^{1-\alpha, i}\right],
\end{align*}
which is a linear function of $\alpha$. Differentiating gives
\begin{align*}\dot{F}_2^i(\alpha) &= \mathbb{E}\Bigg[\int_0^T\!\!\int_0^{t-} G(t-s) \sum_{j\not=i}dX_s^{0,j} d(X^{1,i}-X^{0,i})\\
&\quad\quad\quad - \int_0^T\!\!\int_0^{t-} G(t-s)\sum_{j\not=i}dX_s^{1,j} d(X^{1,i}- X^{0,i})\Bigg]\\
&=- \mathbb{E}\left[\int_0^T\!\!\int_0^{t-} G(t-s) \sum_{j\not=i}d(X_s^{1,j} - X_s^{0,j})d(X^{1,i}_t- X^{0,i}_t)\right],
\end{align*}
which is constant in $\alpha$. Summing over~$i$ yields

\begin{align*}\sum_{i=1}^N\dot{F}_2^i(0)&= - \mathbb{E}\left[\sum_{i=1}^N\sum_{j\not=i}\int_0^T\!\!\int_0^{t-} G(t-s) d(X_s^{1,j} - X_s^{0,j}) d(X^{1,i}_t- X^{0,i}_t)\right]\\
&= \mathbb{E}\Bigg[-\frac{1}{2} \sum_{i=1}^N\sum_{j\not=i}\int_0^T\!\!\int_0^{t-} G(t-s) d(X_s^{1,j} - X_s^{0,j}) d(X^{1,i}_t- X^{0,i}_t)\\
& \quad\quad\quad -\frac{1}{2} \sum_{i=1}^N\sum_{j\not=i}\int_0^T\int_{s+}^T G(t-s)  d(X^{1,i}_t- X^{0,i}_t) d(X_s^{1,j} - X_s^{0,j})\Bigg]\\
&= \mathbb{E}\Bigg[-\frac{1}{2} \sum_{i=1}^N\sum_{j\not=i}\int_0^T\!\!\int_0^T G(|t-s|) d(X_s^{1,j} - X_s^{0,j}) d(X^{1,i}_t- X^{0,i}_t)\\
&\quad\quad\quad +\frac{G(0)}{2}\sum_{i=1}^N\sum_{j\not=i}\sum_{t\in[0,T]}\Delta(X_t^{1,j} - X_t^{0,j})\Delta(X_t^{1,i} - X_t^{0,i})\Bigg].
\end{align*}
As $G$ is positive definite (in fact, here we only use semi-definiteness),
\[-\sum_{i}\sum_{j\not=i}\int_0^T\!\!\int_0^TG(|t-s|)dM_s^jdM_t^i\leq \sum_{i}\int_0^T\!\!\int_0^TG(|t-s|)dM_s^idM_t^i\] for arbitrary Lebesgue--Stieltjes integrators $M^i$. Applying this,
\begin{align}\label{eqn:J2.sum.bd}
\sum_{i=1}^N\dot{F}_2^i(0)&\leq  \mathbb{E}\Bigg[\frac{1}{2} \sum_{i=1}^N\int_0^T\!\!\int_0^T G(|t-s|) d(X_s^{1,i} - X_s^{0,i}) d(X^{1,i}_t- X^{0,i}_t)\\
&\quad\quad\quad +\frac{G(0)}{2}\sum_{i=1}^N\sum_{j\not=i}\sum_{t\in[0,T]}\Delta(X_t^{1,j} - X_t^{0,j})\Delta(X_t^{1,i} - X_t^{0,i})\Bigg].\nonumber
\end{align}
Turning to the jump terms, let
\begin{align*}F_3^i(\alpha)&=\mathbb{E}\left[\frac{G(0)}{2}\sum_{j\not=i}\sum_{t\in[0,T]} \Delta X_t ^{\alpha, i} \Delta X_t^{0,j}+\frac{G(0)}{2}\sum_{j\not=i}\sum_{t\in[0,T]} \Delta X_t ^{1-\alpha, i} \Delta X_t^{1,j}\right].
\end{align*}
This is once again linear in $\alpha$ and differentiating gives
\begin{align*}\dot{F}_3^i(\alpha)&=\mathbb{E}\left[-\frac{G(0)}{2}\sum_{j\not=i}\sum_{t\in[0,T]} \Delta ( X_t^{0,i}-X_t ^{1, i})\Delta X_t^{0,j}+\frac{G(0)}{2}\sum_{j\not=i}\sum_{t\in[0,T]} \Delta (X_t^{0,i}-X_t ^{1, i}) \Delta X_t^{1,j}\right].
\end{align*}
Collecting terms and summing over $i$,
\begin{equation}\label{eqn:J3.sum.bd}\sum_{i=1}^N\dot{F}_3^i(0)=-\mathbb{E}\left[\frac{G(0)}{2}\sum_{i=1}^N\sum_{j\not=i}\sum_{t\in[0,T]} \Delta (X_t^{1,i}-X_t^{0, i})\Delta  (X_t^{1,j}-X_t^{0,j})\right].
\end{equation}

Lastly, we treat the terms arising from the cost $C$. %
Let
\begin{align*}F_4^i(\alpha) &=\mathbb{E}\left[C(X^{\alpha,i})+C(X^{1-\alpha,i})\right].
\end{align*}
By construction, $[0,1]\ni\alpha \mapsto F_4^i(\alpha)$ is a finite convex function  and so its right derivative exists at $0$ (we allow for the possibility that $\dot{F_4^i}(0+)=-\infty$). Moreover, $F_4^i(\alpha)$ is symmetric about $\alpha=1/2$  which together with convexity implies that $\alpha=1/2$ is a minimum. Hence, $0\in\partial F_4^i(1/2)$ where $\partial F_4^i(\alpha)$ denotes the subgradient set at $\alpha$. 
By convexity of $F_4^i$ and the definition of the subgradient, $\dot{F_4^i}(0+)\leq y$ for any $y\in \partial F_4^i(1/2)$. Taking $y=0$ and summing over $i$ we conclude
\begin{align}\sum_{i=1}^N\dot{F}_4^i(0+)
\leq 0. \label{eqn:J4.sum.bd}
\end{align}

Aggregating the above expressions we recover
\[V(\alpha) = \sum_{i=1}^N \left(F_1^i(\alpha)+F_2^i(\alpha)+F_3^i(\alpha)+F_4^i(\alpha)\right)\]
and obtain its right derivative at $0$,
\begin{align*}\dot{V}(0+) &= \sum_{i=1}^N \left(\dot{F}_1^i(0)+\dot{F}_2^i(0)+\dot{F}_3^i(0)+\dot{F}_4^i(0+)\right)\\
&\leq -\mathbb{E}\left[\frac{1}{2}\sum_{i=1}^N\int_0^T\!\!\int_0^T G(|t-s|) d(X_s^{1,i} -X_s^{0, i})d(X_t^{1,i} -X_t^{0, i})\right]<0.
\end{align*}
Here the first inequality follows from \cref{eqn:J1.sum,eqn:J2.sum.bd,eqn:J3.sum.bd,eqn:J4.sum.bd} whereas the last inequality holds as $G$ is strictly positive definite and $\boldsymbol{X}^0\not=\boldsymbol{X}^1$. As mentioned above, $\dot{V}(0+)<0$ is a contradiction and completes the proof.
\end{proof}

\subsection{Proofs for Section \ref{se:singular} (Singular Kernels)}\label{se:proofs.singular}

In this subsection, we extend our results to the case of a singular kernel, where ``singular kernel'' will refer to a function~$G$ satisfying \cref{as:singular.kernel}. Apart from technical verifications, the extension is based on the fact that a singular kernel admits a monotone approximation by kernels satisfying \cref{as:kernel}; cf.\ \cref{rk:kernel.approx}.

\subsubsection{Proof of \cref{le:singular.posdef}}

\begin{proof}[Proof of \cref{le:singular.posdef}]
  We need to show that
  \[
  \int_0^T\!\!\int_0^T G(|t-s|) dX_s dX_t >0
  \]
  whenever $X:[0,T]\to \mathbb{R}$ is a c\`adl\`ag function of finite variation that is not constant and satisfies $\int_0^T\!\!\int_0^T G(|t-s|) d|X|_s d|X|_t<\infty$. 
  By \cite[Proposition~4.5]{GatheralSchiedSlynko.12}, the double integral admits a representation of the form
  \begin{equation}\label{eq:singular.fourier}
  \int_0^T\!\!\int_0^T G(|t-s|) dX_s dX_t=\int|\widehat X(z)|^{2} \mu(dz).
  \end{equation}
  As $X$ has finite variation, the Fourier--Stieltjes transform~$\widehat X(z)=\int_0^Te^{itz}dX_t$ is a continuous function that does not vanish identically as soon as~$X$ is not constant. Moreover, \cite[Lemma~4.2]{GatheralSchiedSlynko.12} gives an expression for the positive Radon measure~$\mu$ from which we see that~$\mu$ has full support. As a consequence, the right hand side of \eqref{eq:singular.fourier} is strictly positive.
\end{proof}

\subsubsection{Objective Function (\cref{prop:singular.obj.func.rep})}

We first record a simple lemma for later of reference. Recall \cref{def:admissible.X.singular} of admissibility.

\begin{lemma}\label{lem:cross.term.integrable}
    For a profile $\boldsymbol{X}=(X^{1},\dots,X^{N})$ of admissible strategies and $i,j\in\{1,\dots,N\}$, \begin{align*}\mathbb{E}&\left[\int_0^\infty\!\!\int_0^\infty G(|t-s|) d|X^j|_sd|X^i|_t\right]\\
    &\quad \quad \quad \quad \quad \leq \frac{1}{2}\mathbb{E}\left[\int_0^\infty\!\!\int_0^\infty G(|t-s|) \left(d|X^i|_sd|X^i|_t+d|X^j|_sd|X^j|_t\right)\right]<\infty.
    \end{align*}
\end{lemma}

\begin{proof}
Using the positive definiteness of $G$,
\begin{align*}0&\leq \mathbb{E}\left[\int_0^\infty\!\!\int_0^\infty G(|t-s|) d(|X^i|_s-|X^j|_s)d(|X^i|_t-|X^j|_t)\right]\\
&=\mathbb{E}\left[\int_0^\infty\!\!\int_0^\infty G(|t-s|) \left(d|X^i|_sd|X^i|_t+d|X^j|_sd|X^j|_t-2d|X^i|_sd|X^j|_t\right)\right].
\end{align*}
Rearranging and applying the admissibility of $X^i$ and $X^j$ yields the claim.
\end{proof}

\begin{proof}[Proof of \cref{prop:singular.obj.func.rep}]
  Recalling the absence of jumps in the present setting (\cref{rk:no.jumps}), the proof follows the same lines as the one of \cref{prop:obj.func.rep}; we merely indicate where additional justifications are needed. 
  First, to justify the application of Fubini's theorem below~\eqref{eqn:obj.intermediate.rep}, we observe that 
\begin{align*}\left|\int_0^T\!\!\int_0^{t-} G(t-s) \sum_{j=1}^NdX_s^{j} dX_t^{i}\right|&\leq \int_0^\infty\!\!\int_0^{\infty} G(|t-s|) \sum_{j=1}^Nd|X^{j}|_sd|X^{i}|_t
\end{align*}
where the right hand side is integrable and therefore almost surely finite by \cref{lem:cross.term.integrable}. 
  Second, to argue $|J(X^{i};\boldsymbol{X}^{-i})|<\infty$ at the end of the proof, note that \eqref{eqn:singular.obj.schied.rep} gives
\begin{align*}
    &|J(X^{i};\boldsymbol{X}^{-i})|\\
    &\leq \mathbb{E}\Bigg[\frac{1}{2}\int_0^\infty\!\!\int_0^\infty G(|t-s|)d|X^{i}|_s d|X^{i}|_t + \int_0^\infty\!\!\int_0^{\infty} G(|t-s|) \sum_{j\not=i}d|X^{j}|_s d|X^{i}|_t +C(X^{i})\Bigg]%
\end{align*}
  which is finite by admissibility, \cref{lem:cross.term.integrable} and our assumption that $\mathbb{E}[C(X^{i})]<\infty$.
\end{proof}

\subsubsection{Extension of \cref{lem:exp.admissible}}\label{se:ext.adm}

Let $X$ be an admissible strategy for the singular kernel~$G$ (\cref{def:admissible.X.singular}). To extend \cref{lem:exp.admissible} to the case where~$G$ is singular, it suffices to show that $\tilde{\mathbb{E}}[X]$ satisfies the integrability condition \eqref{eq:admissible.integrability}. 

Clearly the processes $t\mapsto TV(X;[0,t])$, $t\mapsto TV(\tilde{\mathbb{E}}[X];[0,t])$ and $t\mapsto \tilde{\mathbb{E}}[TV(X;[0,t])]$ are c\`adl\`ag, increasing, and essentially bounded. (As $X$ is continuous in the present setting, these processes are even continuous.) %
As above, we write $d|X|_t$ to denote integration with respect to the total variation of $X$ and $d|\tilde{\mathbb{E}}[X]|_t$ to denote integration with respect to the total variation of $\tilde{\mathbb{E}}[X]$. Similarly, we write $d\tilde{\mathbb{E}}[|X|_t]$ to denote integration with respect to the measure induced by $t\mapsto \tilde{\mathbb{E}}[TV(X;[0,t])]$.

\begin{lemma}\label{lem:TV.meas.ineq} If $G$ is a kernel satisfying \cref{as:kernel},
    \begin{align*}
    \mathbb{E}\left[\int_0^\infty\!\!\int_0^\infty G(|t-s|)d| \tilde{\mathbb{E}}[X]|_s d | \tilde{\mathbb{E}}[X]|_t\right]&\leq \mathbb{E}\left[\int_0^\infty\!\!\int_0^\infty G(|t-s|)d \tilde{\mathbb{E}}[|X|_s] d  \tilde{\mathbb{E}}[|X|_t]\right]\\
    &\leq \mathbb{E}\left[\int_0^\infty\!\!\int_0^\infty G(|t-s|)d|X|_s d |X|_t\right].
    \end{align*}
\end{lemma}

\begin{proof}
By \cref{lem:TV.ineq} the non-negative measures corresponding to $TV(\tilde{\mathbb{E}}[X];[0,t])$ and $\tilde{\mathbb{E}}[TV(X;[0,t])]$ are ordered.\footnote{In the pointwise sense: two  measures $\mu,\nu$ are called ordered if $\mu(A)\leq\nu(A)$ for every Borel set $A$.} Therefore, as $G$ is non-negative,
\[\int_0^\infty\!\!\int_0^\infty G(|t-s|)d| \tilde{\mathbb{E}}[X]|_s d | \tilde{\mathbb{E}}[X]|_t\leq \int_0^\infty\!\!\int_0^\infty G(|t-s|)d \tilde{\mathbb{E}}[|X|_s] d  \tilde{\mathbb{E}}[|X|_t],\]
from which the first claim follows by taking expectations. For the second inequality we use the positive definiteness of the kernel and apply \cref{lem:jensen.ineq.precursor} to the admissible process $t\mapsto TV(X;[0,t])$, 
    \begin{align*}
    0 &\leq \mathbb{E}\left[\int_0^\infty\!\!\int_0^\infty G(|t-s|) d(|X|_s-\tilde{\mathbb{E}}[|X|_s]) d(|X|_t-\tilde{\mathbb{E}}[|X|_t])\right]\\
    & \quad \quad \quad \quad =\mathbb{E}\left[\int_0^\infty\!\!\int_0^\infty G(|t-s|) d|X|_sd|X|_t -\int_0^\infty\!\!\int_0^\infty G( |t-s|)d\tilde{\mathbb{E}}[|X|_s] d\tilde{\mathbb{E}}[|X|_t]\right].
\end{align*}
Rearranging completes the proof.
\end{proof}

We are now in a position to prove that \eqref{eq:admissible.integrability} holds for $\tilde{\mathbb{E}}[X^i]$.

\begin{lemma}\label{le:exp.G.int} If $X$ is admissible for the singular kernel $G$, then
    \begin{align*}\mathbb{E}\left[\int_0^\infty\!\!\int_0^\infty G(|t-s|)d|\tilde{\mathbb{E}}[X]|_s d |\tilde{\mathbb{E}}[X]|_t\right]&\leq \mathbb{E}\left[\int_0^\infty\!\!\int_0^\infty G(|t-s|)d|X|_s d |X|_t\right]<\infty
    \end{align*}
and therefore $\tilde{\mathbb{E}}[X^i]$ is admissible; i.e., \cref{lem:exp.admissible} holds for the singular kernel $G$.
\end{lemma}

\begin{proof}
    As seen in \cref{rk:kernel.approx}, the singular kernel $G$ can be approximated by a sequence of kernels $G^n\uparrow G$ satisfying \cref{as:kernel}. By the monotonicity of the sequence and the non-negativity of the integrating measures,
    \begin{align*}0\leq \int_0^\infty\!\!\int_0^\infty G^n(|t-s|)d| \tilde{\mathbb{E}}[X]|_s d | \tilde{\mathbb{E}}[X]|_t&\leq \int_0^\infty\!\!\int_0^\infty G^{n+1}(|t-s|)d| \tilde{\mathbb{E}}[X]|_s d | \tilde{\mathbb{E}}[X]|_t, \ \ \ \forall n\geq1.
    \end{align*}
    Given that $X$ is admissible for $G$, by  applying \cref{lem:TV.meas.ineq} for each $G^n$ we have
    \begin{align*}\mathbb{E}\left[\int_0^\infty\!\!\int_0^\infty G^n(|t-s|)d|\tilde{\mathbb{E}}[X]|_s d |\tilde{\mathbb{E}}[X]|_t\right]&\leq \mathbb{E}\left[\int_0^\infty\!\!\int_0^\infty G^n(|t-s|)d|X|_s d |X|_t\right]\\
    &\leq \mathbb{E}\left[\int_0^\infty\!\!\int_0^\infty G(|t-s|)d|X|_s d |X|_t\right]<\infty.    \end{align*}
    Using that $G^n\uparrow G$ for $d|\tilde{\mathbb{E}}[X]|\otimes d|\tilde{\mathbb{E}}[X]|$ almost every $(s,t)$, almost surely, we have by monotone convergence that
    \begin{align*}\mathbb{E}\left[\int_0^\infty\!\!\int_0^\infty G(|t-s|)d|\tilde{\mathbb{E}}[X]|_s d |\tilde{\mathbb{E}}[X]|_t\right]&\leq \mathbb{E}\left[\int_0^\infty\!\!\int_0^\infty G(|t-s|)d|X|_s d |X|_t\right]<\infty. \qedhere
    \end{align*}
\end{proof}

\subsubsection{Extension of \cref{prop:reduced.rep.obj}}\label{se:ext.proj}

The following stability property of the objective with respect to monotone approximation of~$G$ will be used to extend \cref{prop:reduced.rep.obj} to the singular case. When there are several kernels under consideration, we write $J_G$ for the objective function with kernel~$G$.

\begin{proposition}\label{prop:G.consistency}
    Let $\boldsymbol{X}$ be an admissible strategy profile for the singular kernel $G$ and let $(G^n)_{n\geq1}$ be a sequence of kernels satisfying \cref{as:kernel} with $G^n\uparrow G$. Then
    \[\lim_{n\to\infty}J_{G^n}(X^i;\boldsymbol{X}^{-i})=J_G(X^i;\boldsymbol{X}^{-i}).\]
\end{proposition}

\begin{proof}
    As  $\boldsymbol{X}$ is admissible for $G$, it has no jumps and the objective is given by~\eqref{eqn:singular.obj.schied.rep}. Consider $J_{G^n}(X^i;\boldsymbol{X}^{-i})$ as given by~\eqref{eqn:singular.obj.schied.rep} with $G^n$, we want to take the limit and obtain the same expression with~$G$ instead of $G^n$. Indeed, by \cref{lem:cross.term.integrable},
\begin{align*}
    \frac{1}{2}\int_0^T\!\!\int_0^T G(|t-s|)d|X^{i}|_s d|X^{i}|_t + \int_0^T\!\!\int_0^{T} G(|t-s|) \sum_{j\not=i}d|X^{j}|_s d|X^{i}|_t
\end{align*}
is integrable. As $0\leq G^n \leq G$ and $G^n\to G$ pointwise, dominated convergence yields the claim.
\end{proof}

Using this stability result, we readily obtain the extension of \cref{prop:reduced.rep.obj}.

\begin{proposition}\label{prop:ext.proj}
    \Cref{prop:reduced.rep.obj} holds for the singular kernel $G$.
\end{proposition}

\begin{proof}
Given an admissible strategy profile $\boldsymbol{X}$ for the singular kernel $G$, we need to show that $J_G(X^i;\boldsymbol{X}^{-i})=J_G(X^i;\tilde{\mathbb{E}}[\boldsymbol{X}^{-i}])$. Define the approximating sequence $G^n\uparrow G$ as in \cref{rk:kernel.approx}. We apply \cref{prop:G.consistency} to the strategies $X^i,\boldsymbol{X}^{-i}$ and $\tilde{\mathbb{E}}[\boldsymbol{X}^{-i}]$ which are admissible for $G$ by \cref{le:exp.G.int}: using \cref{prop:reduced.rep.obj} for each $G^n$, we obtain
\[J_G(X^i;\boldsymbol{X}^{-i})=\lim_{n\to\infty}J_{G_n}(X^i;\boldsymbol{X}^{-i})=\lim_{n\to\infty}J_{G_n}(X^i;\tilde{\mathbb{E}}[\boldsymbol{X}^{-i}])=J_{G}(X^i;\tilde{\mathbb{E}}[\boldsymbol{X}^{-i}]). \qedhere\]
\end{proof}

\subsubsection{Extension of \cref{thm:non.existence.random.eq}}\label{se:ext.non.exist}

The extension of \cref{thm:non.existence.random.eq} does not follow from the stability result alone (which would only give a non-strict inequality). Instead, we appeal once again to strict positive definiteness.

\begin{proposition}\label{prop:ext.non.exist}
    \Cref{thm:non.existence.random.eq} holds for the singular kernel $G$.
\end{proposition}

\begin{proof}
   In view of~\eqref{eqn:singular.obj.schied.rep}, it suffices to show two (in)equalities,
   \begin{align}\label{eq:proof.ext.non.exist.1}
   \mathbb{E}\left[\int_0^T\!\!\int_0^{t-} G(t-s)\sum_{j\not=i}d\tilde{\mathbb{E}}[X_s^{j}] dX^i_t\right]=\mathbb{E}\left[\int_0^T\!\!\int_0^{t-} G(t-s)\sum_{j\not=i}d\tilde{\mathbb{E}}[X_s^{j}] d\tilde{\mathbb{E}}[X^i_t] \right]
   \end{align}
   and
   \begin{align}\label{eq:proof.ext.non.exist.2}
   \mathbb{E}\left[\frac{1}{2}\int_0^T\!\!\int_0^T G(|t-s|) dX_s^{i} dX^i_t\right]>\mathbb{E}\left[\frac{1}{2}\int_0^T\!\!\int_0^T G(|t-s|) d\tilde{\mathbb{E}}[X_s^{i}] d\tilde{\mathbb{E}}[X^i_t]\right].
   \end{align}
   Define again the approximating sequence $G^n\uparrow G$ as in \cref{rk:kernel.approx}. Then \eqref{eq:proof.ext.non.exist.1} holds with $G^n$ instead of $G$, and since $X^i,X^j$  satisfy the admissibility condition~\eqref{eq:admissible.integrability}, we can repeat the same argument as in \cref{prop:G.consistency} to apply dominated convergence and deduce \eqref{eq:proof.ext.non.exist.1} for~$G$. To prove~\eqref{eq:proof.ext.non.exist.2}, we first argue that the equality
   \begin{align}\label{eq:proof.ext.non.exist.3}
    \mathbb{E}&\left[\frac{1}{2}\int_0^\infty\!\!\int_0^\infty G(|t-s|) d(X^i_s-\tilde{\mathbb{E}}[X^i_s]) d(X^i_t-\tilde{\mathbb{E}}[X^i_t]) \right] \\
    &\quad \quad \quad \quad \quad =\mathbb{E}\left[\frac{1}{2}\int_0^\infty\!\!\int_0^\infty G(|t-s|) dX^i_sdX^i_t -\frac{1}{2}\int_0^\infty\!\!\int_0^\infty G( |t-s|)d\tilde{\mathbb{E}}[X^i_s] d\tilde{\mathbb{E}}[X^i_t]\right] \nonumber
\end{align}
from \cref{lem:jensen.ineq.precursor} extends to singular~$G$. Indeed, \eqref{eq:proof.ext.non.exist.3} holds for each $G^n$ by \cref{lem:jensen.ineq.precursor}, and given that $X^i$ and $\tilde{\mathbb{E}}[X^i]$ satisfy \eqref{eq:admissible.integrability}, dominated convergence then yields \eqref{eq:proof.ext.non.exist.3} for $G$. As $G$ is strictly positive definite (\cref{le:singular.posdef}), our assumption that $X^i-\tilde{\mathbb{E}[X^i]}$ is not identically zero yields
\begin{align*}
    0<\mathbb{E}&\left[\frac{1}{2}\int_0^\infty\!\!\int_0^\infty G(|t-s|) d(X^i_s-\tilde{\mathbb{E}}[X^i_s]) d(X^i_t-\tilde{\mathbb{E}}[X^i_t]) \right].
\end{align*}
Applying~\eqref{eq:proof.ext.non.exist.3} and rearranging gives the desired inequality~\eqref{eq:proof.ext.non.exist.2}.
\end{proof}

\subsubsection{Proof of \cref{thm:extensions.singular}}\label{se:pf.ext.thm.sing}

In the preceding subsections, we have shown the extensions of \cref{lem:exp.admissible}, \cref{prop:reduced.rep.obj} and \cref{thm:non.existence.random.eq}. %
The extension of the uniqueness result (\cref{thm:nash.eq.unique}) is obtained by following the steps in the proof of \cref{thm:nash.eq.unique} and omitting the jump terms. To justify the applications of Fubini's theorem, 
we use \cref{lem:cross.term.integrable} as in the proof of \cref{prop:singular.obj.func.rep}.

\bibliographystyle{abbrv}
\bibliography{stochfin}

\end{document}